\documentclass[12pt,reqno]{article}

\usepackage{amsmath,amsfonts,amsthm,amssymb,amsxtra}
\usepackage{tikz}
\usepackage{bbm} 
\usepackage[colorlinks,citecolor=red,pagebackref,hypertexnames=false]{hyperref} 

\usepackage[makeroom]{cancel}
\usepackage{breqn}

\newtheorem{theorem}{Theorem}

\newtheorem{lemma}[theorem]{Lemma}

\theoremstyle{definition}

\newtheorem{definition}[theorem]{Definition}

\theoremstyle{remark}

\newtheorem{remark}{Remark}

\newcommand{\C}{\mathbb{C}}

\renewcommand{\epsilon}{\varepsilon}

\renewcommand{\phi}{\varphi}
\newcommand{\R}{\mathbb{R}}

\newcommand{\Z}{\mathbb{Z}}
\newcommand{\om}{\omega}
\newcommand{\Om}{\Omega}
\newcommand{\lat}{\mathcal{L}}
\newcommand{\E}{\mathbb{E}}
\newcommand{\lan}{\langle}
\newcommand{\ran}{\rangle}
\newcommand{\g}{\gamma}

\newcommand{\vertiii}[1]{{\left\vert\kern-0.25ex\left\vert\kern-0.25ex\left\vert #1 
    \right\vert\kern-0.25ex\right\vert\kern-0.25ex\right\vert}}

\newcommand{\eps}{\varepsilon}

\usepackage{marginnote}

\DeclareMathOperator{\Tr}{Tr}

\DeclareMathOperator{\den}{den}

\title{On the reduced Hartree-Fock equations with a  Small  Anderson Type Background Charge Distribution}

\author{Ilias Chenn and Shiwen Zhang}
\newcommand{\Addresses}{{
  \bigskip
  \vskip 0.08in \noindent --------------------------------------

  \footnotesize

\medskip

I.~Chenn, \textsc{Department of Mathematics, Massachusetts Institute of Technology, 2-252b, 77 Massachusetts Avenue, 
Cambridge, MA  4307 USA }\par\nopagebreak
  \textit{E-mail address}: \texttt{nehcili@mit.edu}

\medskip

S.~Zhang, \textsc{School of Mathematics, University of Minnesota, 206 Church St SE, Minneapolis, MN 55455 USA}\par\nopagebreak
  \textit{E-mail address}: \texttt{zhan7294@umn.edu }
  
}}

\date{}

\begin{document}

\maketitle

\begin{abstract}
We demonstrate that the reduced Hartree-Fock equation (REHF) with  a small  Anderson type background charge distribution has an unique stationary solution by explicitly computing a screening mass at positive temperature.
\end{abstract}

\section{Introduction}
Density functional theory (DFT) has become the stable of modern quantum science thanks to its tractability. One of its notable application is in describing the electronic structure of disordered crystals in solid state physics and material sciences. Nevertheless, the disordered nature of certain crystal structures and the long range Coulomb interaction have caused considerable difficulties in their theoretical studies. To illustrate this difficulty, we restrict our attention to dimension three (3) for the rest of the paper. 

 In condensed matter physics, the density of electrons is governed by the Kohn-Sham (KS) equation of DFT, whose mathematical properties have become an area of intense research, for example, \cite{AACan, CDL,CLL,CSL, CLeBL, CLeBL2, CS, CS1, ELu, Nier, La, Lev, Levy, Levy2, Lieb3, Lieb, LS,PN}. These works consider various problems in KS model, e.g., \cite{AACan} is concerned with KS models in the local density approximation framework, \cite{CDL,CLL,CSL, La, LS} consider  interactions in crystals with local defects or systems with  infinite particles, the works \cite{CLeBL, CLeBL2} consider thermodynamic limit for a periodic model. The literature on KS model is huge, and we are not able to give a thorough review here. We refer the readers to, e.g.,  \cite{ELu} and reference therein for more details.   One central focus of these studies is the existence and uniqueness problem. In this context, a prolific approach  is through the use of variational arguments. Given a nuclear charge distribution $\kappa$ on a set $S \subset \R^3$, the Kohn-Sham equation are the Euler-Lagrange equation of the Kohn-Sham energy, which is a functional of the electron density matrix $\gamma$ \cite{KS}: 
\begin{align}
    E_{\rm KS}(\gamma) =& \Tr((-\Delta)\gamma) \notag \\
    &+ \frac{1}{2} \int_{ {\R^3\times \R^3}} dxdy \, \frac{(\kappa(x)-\rho_\gamma(x))(\kappa(y)-\rho_\gamma(y))}{|x-y|} \label{eqn:Coulumb-bad-term} \\
    &+ Ex(\rho_\gamma), \notag
\end{align}
where $\rho_\gamma(x) = \gamma(x,x)$ is the electron density ($\gamma(x,y)$ is the integral kernel of $\gamma$) and $Ex$ is the exchange-correlation energy depending solely on $\rho_\gamma$, instead of $\gamma$.  However, a major impediment to the variational approach is that  the Coulomb self-interaction energy in \eqref{eqn:Coulumb-bad-term} becomes ill defined if $S$ is unbounded and $\kappa - \rho_\gamma$ lacks decay.  A notable example is when $\kappa$ is the Anderson type disorder. 

At zero temperature,  some mitigation of the long range Coulomb problem involves  studying minimizers of $E_{\rm KS}$ on a finite domain and  then passing to their  thermodynamic  limits to include periodic $\kappa$ and $\rho_\gamma$ (see for example, \cite{CLeBL, CLeBL2, CS1}), or introducing a screening mass as in the case of Yukawa potential \cite{CLL}. On the other hand, a natural screening is observed at positive temperature \cite{CS,Lev}. 
    
In this work, we follow the latter school of thoughts. We show that the Kohn-Sham equation have unique solutions despite the presence of long range Coulomb interaction and disordered background distribution $\kappa$ at positive temperature. To elucidate the main ideas, we drop the exchange correlation term $Ex$ in $E_{\rm KS}$. The resulting energy and equation bear the name of  reduced Hartree-Fock (REHF) equation and energy, respectively. Let $\beta$ denote the inverse temperature. At positive temperature $\beta^{-1}$, the REHF free energy of $\gamma$  is
\begin{align}
    F_{\rm REHF}(\gamma) &:= E_{\rm REHF}(\gamma) - \beta^{-1}S(\gamma), \\
    E_{\rm REHF}(\gamma) &:= \Tr((-\Delta)\gamma) \notag +\frac{1}{2} \int_{ {\R^3\times \R^3}} dxdy \, \frac{(\kappa(x)-\rho_\gamma(x))(\kappa(y)-\rho_\gamma(y))}{|x-y|}, \\
    S(\gamma) &:= -\Tr(\gamma \ln(\gamma) + (1-\gamma)\ln(1-\gamma)),
\end{align}
whenever $F_{\rm REHF}$ is well defined. We remark that we have not specified the space on which $\gamma$ acts as well as the exact meaning of $\Tr$. Since we will be working with and describing the Euler-Lagrange equations later, this inconvenience is only formal. The REHF equation is the Euler-Lagrange equation of $F_{\rm REHF}$ with a fixed number of particles $\Tr \gamma$.  Despite the fact that $F_{\rm REHF}$ is a functional of $\gamma$, the REHF equation can be solely written in terms of the electron density $\rho$ (hence the name DFT), the (auxiliary) electric potential $\varphi$, and the chemical potential $\mu$ (i.e. the Lagrange-multiplier):
\begin{align}
    \rho &= \den f_{\rm FD}(\beta(-\Delta -\phi - \mu)), \label{eqn:rho-eqn-1}\\
    -\Delta \phi &= 4\pi \big(\kappa - \rho \big), \label{eqn:sub-rho-in-phi}
\end{align}
where $f_{\rm FD}$ is the Fermi-Dirac distribution
\begin{equation}
    f_{\rm FD}(x) = \frac{1}{1+e^x}.
\end{equation}
Moreover, we have introduced the $\den$ operator mapping operators $A$ on $L^2(\R^3)$ to its density $A(x,x)$ where  $A(x,y)$ is its integral kernel. A more precise definition of $\den$ can be found in Appendix \ref{app:per-vol-set-up} or \cite{CLL}. Nevertheless, one can recover $\gamma$ via
\begin{equation}
    \gamma = f_{\rm FD}(\beta(-\Delta -\phi - \mu)).
\end{equation}
However, a yet more advantageous view is to substitute $\rho$ from \eqref{eqn:rho-eqn-1} into \eqref{eqn:sub-rho-in-phi} to obtain
\begin{equation}
    -\frac{1}{4\pi}\Delta \phi = \kappa - \den f_{\rm FD}(\beta(-\Delta -\phi - \mu)). \label{eqn:REHF}
\end{equation}
It is in this view that one can see a clearer presence of a screening mass at positive temperature. We will henceforth call equation \eqref{eqn:REHF} the reduced Hartree-Fock (REHF) equation for the purpose of this paper.

To specify the class of disorders in $\kappa$ and the solution space on which we study \eqref{eqn:REHF}, let $\lat$ denote a (non-degenerate) Bravais lattice in $\R^3$ and $Q$ denote any fundamental domain. That is, $\lat$ translates of $Q$ tile $\R^3$. A typical example of $Q$ is the Wigner-Seitz cell. Moreover,  assume that $\tau=\{\tau_\ell\}_{\ell \in \lat}$ is a group of 
measure preserving $\lat$-actions on the probability space $\Om$.

\begin{definition} \label{def:stationary-func}
A measurable function $f$ on $\R^3 \times \Om$ is said to be ($\lat$) stationary if
\begin{equation}\label{eq:lat}
	f(x-\ell, \tau_\ell \om) = f(x,\om)
\end{equation}
for all $(x,\om) \in \R^3 \times \Om$ and $\ell \in \lat$.
\end{definition}

A notable example of this category  is the  Anderson potential 
\begin{equation}
    \kappa(x,\omega)=\sum_{\ell \in \lat=\Z^d}q_{\ell}(\omega)\chi(x-\ell),
\end{equation}
where $\chi\in C_c^\infty(\R^d)$ and the $q_{\ell}(\omega)$ are i.i.d. random variables. The Anderson type potential is of 
interest to physicists and mathematicians alike \cite{CL,Kir}.


 We define the space of stationary $\kappa$'s and record the conventions used before we state our main result.  Let $L^q(Q)$ denote the usual $L^q$ space over $Q$ with the standard $L^q$ norm and $L^2$ inner products when $q=2$. We will use the notation $L^p_\om L^q_x$ to denote the norm-completed
\begin{equation}
	\{ f(x,w) \in L_{\rm loc}^q(\R^3) \times L^p(\Om) : \text{$f$ is stationary and } \|f\|_{L^p_\om L^q_x} < \infty \},
\end{equation}
where 
\begin{equation}
	\|f\|_{L^p_\om L^q_x}^p := \E \|f\|_{L^q(Q)}^p.
\end{equation}
When $p=\infty$ or $q=\infty$, the associated norm is the usual sup norm. We remark that due to egodicity, the above norms do not depend on the location of the fundamental domain $Q$  (see Lemma \ref{lem:indep-of-loc-of-Q} in Appendix \ref{app:per-vol-set-up}). Consequently, we do not display $Q$ in the subscripts of the norms. 

With these notations, equation \eqref{eqn:REHF} is to be understood as follows. We require $\kappa \in L^\infty_\om L^2_x$ while we look for a solution $(\varphi,\mu)$ in $L^\infty_\om H^2_x\times \R$. Nevertheless, $-\Delta -\phi - \mu$ and $f_{\rm FD}(\beta(-\Delta -\phi - \mu))$ are operators on $L^2(\R^3)$, where the latter is defined by the holomorphic functional calculus. Consequently, $\den f_{\rm FD}(\beta(-\Delta -\phi - \mu))$ is a stationary function since $\phi$ is stationary.

Finally, we will use 
\[
    \hat f(k)=(2\pi)^{-3/2}\int_{\R^3} f(x)e^{-i x\cdot k} dx
\]
as our unitary Fourier transform.  Moreover, we will denote by $C_{a,b,\cdots}$ finite constants depending only on the parameters $a,b,\cdots$. By an abuse of notation and by choosing an even large constant if necessary, $C$ or $C_{a,b,\cdots}$ may stand for different constants simultaneously. Lastly, we will write $ A\lesssim B, A\gtrsim B$ if $A\le CB$ or $A\ge CB$ for some absolute constant $C$, respectively.

\subsection{Main results}

We notice that REHF exhibits a symmetry
\begin{equation}\label{eq:sym}
    (\varphi,\mu)\mapsto (\varphi+t,\mu-t)
\end{equation}
for any $t\in\R$. That is, 
if  $(\varphi,\mu)$ is  a solution of \eqref{eqn:REHF}, then so is $(\varphi+t,\mu-t)$. Therefore, a solution of REHF is an equivalence class of $(\varphi,\mu)\in L^\infty_\om H^2_x\times \R$ with respect to \eqref{eq:sym}. Consequently, the notion of uniqueness is defined up to this equivalence class.  With this definition of uniqueness, we state our main results. 

Let $Q$ denote a fundamental domain of the lattice $\lat$.  Let $\kappa \in L^\infty_w  L_x^2 $   and $\kappa_0$ denote its expected spacial average:
\begin{equation}
   \kappa_0 = \frac{1}{|Q|} \E \int_Q \kappa.
\end{equation}
Let $\kappa' = \kappa - \kappa_0$. We have the following main results below. Theorem \ref{thm:trans-invar-REHF} considers the simplest case where $\kappa \in \R$ is homogeneous. That is, the material is a jellium. Theorem \ref{thm:main-result} extends Theorem \ref{thm:trans-invar-REHF} by considering random disordered background $\kappa$, as a perturbation from the jellium solution with background potential $\kappa_0$.

\begin{theorem} \label{thm:trans-invar-REHF}
Let $0 < \kappa \in \R$. There exists positive constants $c_1=\dfrac{1}{8\pi^{3/2}}$ and $c_2=\dfrac{1}{12\pi^2}$ such that if $\beta \in \R$ is positive and satisfies
\begin{equation}
    \kappa > \frac{c_1}{\beta^{3/2}}, \label{eqn:trans-invar-ass}
\end{equation}
then there exists a solution $(0, \mu) \in \R \times \R$ to the REHF equation \eqref{eqn:REHF}. Moreover,
\begin{equation}
    0 < \frac{1}{\beta}\log(\kappa\beta^{3/2}/c_1) < \mu < (\kappa/c_2)^{2/3}. \label{eqn:mu-bounds-free-case}
\end{equation}
\end{theorem}
 
 \begin{remark}
 We will see in the proof that the density of the homogeneous Poisson equation \eqref{eqn:REHF} is a continuous increasing function of $\mu$. 
 We request the  technical condition \eqref{eqn:trans-invar-ass} in order to obtain a positive solution $\mu$. See equations \eqref{eqn:free-mu-upper}, \eqref{eqn:free-mu-lower} in the proof of Theorem \ref{thm:trans-invar-REHF} for more details.
 \end{remark}
 
\begin{theorem} \label{thm:main-result}
  Let  $0<\beta<\infty$ and $\kappa \in L^\infty_w  L_x^2$ be real valued. Assume that \eqref{eqn:trans-invar-ass} holds with $\kappa_0$ and let $\mu>0$ be given as in Theorem \ref{thm:trans-invar-REHF}. There exist constant $C_1,C_2>0$ depending on $\beta,\mu$ and $\kappa_0$ , see \eqref{eq:kappa-small},\eqref{eq:small-2}. If $\|\kappa'\|_{L^\infty_w L_x^2}\le C_1$, then the REHF equation \eqref{eqn:REHF} has a unique solution $(\phi, \mu) \in  L^\infty_w H^2_x \times \R$ in a neighborhood of $(0, \mu)$ with
\begin{equation} \label{eq:thm3-small}
	\| \phi\|_{L^\infty_w H^2_x} \le \,  C_{2}\|\kappa'\|_{L^\infty_w L_x^2}.
\end{equation} 
\end{theorem}

\begin{remark}
The positive assumption $\beta^{-1}>0$ is crucial for our main results. For technical reasons, the smallness condition on $\kappa'$ and the constant $C_2$ in \eqref{eq:thm3-small} decays to 0 as $\beta\to\infty$. Our approach does not apply to the zero temperature case. At the zero temperature, the existence of the solution to the reduced Hartree-Fock equations in disordered media have been studied and established in the previous work \cite{CLL,La}, when the interaction is short-range. See Proposition 4.5 in \cite{CLL} and Theorem 2.1 in \cite{La}. We refer readers to  these work and references therein for more related results. 
\end{remark}

By our remark immediately after Definition \ref{def:stationary-func}, we see that Theorem \ref{thm:main-result} provides existence and uniqueness results for  small  Anderson
 background charge distributions.

The paper is organized as follows. In section \ref{sec:trans-invar-soln}, we provide a proof for Theorem \ref{thm:trans-invar-REHF} in the case $\kappa = \kappa_0$ is a constant. Then, we prove Theorem \ref{thm:main-result} perturbatively by linearizing the REHF equation at the constant solution to the case $\kappa = \kappa_0$. 
The main proof of Theorem \ref{thm:main-result} is given in Section \ref{sec:main-proof}, modulo the core linear and nonlinear analysis. The linearized operator of the REHF equation is studied in Section \ref{sec:lin-anal} while the nonlinear analysis is given in Section \ref{sec:nonlin-anal}.

\vspace{1cm}

\noindent
\textbf{Acknowledgments.} We thank Svitlana Mayboroda for suggesting to us the topic of the current article. We would also like to thank Wei Wang for stimulating discussions. Chenn is grateful for I. M. Sigal for his insights and guidance.

Chenn is supported through a Simons Foundation Grant (601948 DJ) and a PDF fellowship from NSERC/Cette recherche a \'{e}t\'{e} financ\'{e}e par le CRSNG.  Zhang
is supported in part by the NSF grants DMS1344235, DMS-1839077, and Simons Foundation grant 563916, SM.

\section{Translation invariant solution: proof of Theorem \ref{sec:trans-invar-soln}} \label{sec:trans-invar-soln}

\begin{proof}[Proof of Theorem \ref{thm:trans-invar-REHF}]
If $\kappa \in \R$ and $\phi = 0$, then equation \eqref{eqn:REHF} becomes  an equation for $\mu$: 
\begin{equation}
  \kappa = \den f_{\rm FD}(\beta(-\Delta - \mu)). 
\end{equation}
Translation invariance of $-\Delta$ on $L^2(\R^3)$ shows that this equation can be rewritten as
\begin{equation}
    \kappa = \frac{1}{(2\pi)^{3 }}\int_{\R^3} dp\, f_{\rm FD}(\beta(p^2-\mu)).  \label{eqn:translation-invariant-REHF}
\end{equation}
Thus, the right hand side of \eqref{eqn:translation-invariant-REHF} becomes
\begin{equation}
 A(\mu):=    \frac{1}{(2\pi)^{3 }}4\pi \int_0^\infty dq \frac{q^2}{1+e^{-\beta \mu} e^{\beta q^2}}. \label{eqn:RHS-mu}
\end{equation}
An upper bound for $A(\mu)$ in \eqref{eqn:RHS-mu} is
\begin{align}
   \frac{1}{2\pi^2} \int_0^\infty dq\, q^2 e^{\beta \mu} e^{-\beta q^2} =& \frac{1}{2\pi^2} \,e^{\beta \mu} \beta^{- 3/2}  \, \int_0^\infty\, dq\, q^2 e^{-q^2} \\
    =& \ \frac{1}{8\pi^{3/2}}  e^{\beta \mu} \beta^{- 3 /2}:=B(\mu). \label{eqn:free-mu-upper}
\end{align}
As $\mu$ ranges from $0$ to $\infty$, we see that $B(\mu)$ in  \eqref{eqn:free-mu-upper} increases from $c_1 \beta^{-3/2}$ to $\infty$ for $c_1=\dfrac{1}{8\pi^{3/2}}$:
\[c_1 \beta^{-3/2}=B(0)\le B(\mu)\le B(\infty)=\infty.\]
Similarly, a lower bound for \eqref{eqn:RHS-mu} is
\begin{equation}
    \frac{1}{2\pi^2} \int_0^{\sqrt{\mu}} dq\,  \frac{q^2}{2} =  \frac{1}{12\pi^2} \mu^{3/2}:=C(\mu). \label{eqn:free-mu-lower}
\end{equation}
As $\mu$ ranges from $0$ to $\infty$,  the lower bound estimate $C(\mu)$ in  \eqref{eqn:free-mu-lower} increases from $0$ to $\infty$:
\[0=C(0)\le C(\mu)\le C(\infty)=\infty.\]
Thus, continuity (in $\mu$) of the right hand side of equation \eqref{eqn:translation-invariant-REHF} implies that a solution $\mu > 0$ exists if $\kappa>c_1 \beta^{-3/2}$. Solving for $\mu$ in \eqref{eqn:free-mu-upper} and \eqref{eqn:free-mu-lower}, i.e., $C(\mu)< \kappa=A(\mu)< B(\mu)$ implies
\[\frac{1}{12\pi^2} \mu^{3/2}<\kappa < \frac{1}{8\pi^{3/2}}  e^{\beta \mu} \beta^{- 3 /2}.\]
Therefore,
\[\frac{1}{\beta}\log\big(8\pi^{3/2}\kappa \beta^{3/2}\big) < \mu < (12\pi^2\kappa)^{2/3}\]
proves \eqref{eqn:mu-bounds-free-case} with $c_1=\dfrac{1}{8\pi^{3/2}}$ and $c_2=\dfrac{1}{12\pi^2}$.
\end{proof}

\section{Fixed point argument: proof of Theorem \ref{thm:main-result}} \label{sec:main-proof}
\begin{proof}[Proof of Theorem \ref{thm:main-result}]
By Theorem \ref{thm:trans-invar-REHF} and  the assumption \eqref{eqn:trans-invar-ass}  on $\kappa_0$, there exists a solution $(\phi=0, \mu)$ with $\mu > 0$ to equation \eqref{eqn:REHF} with 
\begin{equation}
    \kappa = \kappa_0 := \frac{1}{|Q|}\E\int_Q \kappa. \label{eqn:kappa-0-def}
\end{equation}
That is, $\kappa_0=\den f_{\rm FD}(\beta(-\Delta - \mu))$. 
For this $\mu$ fixed, we look for a solution $(\phi, \mu)$ to the full equation \eqref{eqn:REHF} perturbatively. In particular, we linearize \eqref{eqn:REHF} at $\phi = 0$. Let
\begin{equation}
	M := d_\phi \den f_{\rm FD}(\beta(-\Delta - \phi - \mu)) \mid_{\phi = 0} \label{eqn:linop}
\end{equation}
and
\begin{equation}
	N(\phi) :=  -\Big(\den f_{\rm FD}(\beta(-\Delta - \phi - \mu)) - \den f_{\rm FD} (\beta(-\Delta - \mu)) - M\phi\Big)  . \label{eqn:nonlinDef}
\end{equation}
Then we can rewrite equation \eqref{eqn:REHF} as
\begin{equation}
    (-\frac{1}{4\pi}\Delta + M)\phi = \kappa' + N(\phi). \label{eqn:LS-standard-form}
\end{equation}
Let us further denote $L$ to be the linear operator
\begin{equation}
    L := -\frac{1}{4\pi}\Delta + M.
\end{equation}
We have the following estimates for the linear operator $L$ and the nonlinear operator $N$, whose proofs are delayed to sections \ref{sec:lin-anal} and \ref{sec:nonlin-anal} below for a self-contained treatment. 

\begin{theorem} \label{thm:LLowerBound-final}
Let $L$ be defined in \eqref{eqn:linop} on $L^\infty_w L^2_x$. Let $f \in L^\infty_w L^2_x$. Then
\begin{equation}
    \|Lf\|_{L^\infty_w L^2_x} \geq  \|(m_Q-\frac{1}{8\pi}\Delta)f\|_{L^\infty_w L^2_x} 
\end{equation}
where
\begin{equation}
    m_Q =c_\ast \min(\mu, \sqrt{\mu})  - C_{\beta,\mu}\,  \ell(Q)^{-1} \label{eqn:m-Q-def}
\end{equation}
for some universal constant $c_\ast>0$ and some constant $C_{\beta,\mu} > 0$  depending only on $\beta,\mu$, and $\ell(Q)$ is the diameter of $Q$. 
\end{theorem}

\begin{theorem} \label{thm:nonlin}
Let $N$ be defined in \eqref{eqn:nonlinDef} and let $\phi_1, \phi_2 \in L^\infty_w H^2_x$. There exists a constant $C_{\beta,\mu}$ depending only on $\beta$ and $\mu$ such that if  $\|\phi_i \|_{L^\infty_\om H_x^2} \leq \frac{1}{10}C_{\beta,\mu}^{-1}$,  then
\begin{equation}
	\|N(\phi_1) - N(\phi_2)\|_{L^\infty_w L^2_x} \leq C_{\beta,\mu}(\|\phi_1\|_{L^\infty_w H^2_x} + \|\phi_2\|_{L^\infty_w H^2_x}) \|\phi_1 - \phi_2\|_{L^\infty_w H^2_x}. \label{eqn:nonlin}
\end{equation}
\end{theorem}

 In order to utilize Theorem \ref{thm:LLowerBound-final} and ensure a positive lower bound for $L$, we make the following observation. If $\kappa$ is $\lat$ stationary, it is also $(n\lat)$ stationary for $0 < n \in \Z$. 
  Given $Q$, if its diameter $\ell(Q)$ is not large enough to ensure that $m_Q$ in \eqref{eqn:m-Q-def} is positive, we may then work on an enlarged fundamental domain $nQ$ first for some large positive integer $n$ to be specified. By the definition of $\kappa_0$ in \eqref{eqn:kappa-0-def} and the stationary of $\kappa$, $\kappa_0$ remains the same either on $Q$ or $nQ$. We obtain the same $\mu$ (depending only on $\kappa_0$) on $L^\infty_\om  L^2_x(nQ)$ by Theorem \ref{thm:trans-invar-REHF}. We then apply Theorem \ref{thm:LLowerBound-final} to $f\in L^\infty_\om  L^2_x(nQ)$, where the associated lower bound is
 \begin{equation}
     m_{nQ} =c_\ast \min(\mu, \sqrt{\mu})  - C_{\beta,\mu}\,  \ell(nQ)^{-1} .
 \end{equation}
 We then choose $n$ so large (depending only on $\mu,\beta$) so that $m_{nQ}>0$. Theorem \ref{thm:nonlin} holds on $L^\infty_\om H_x^2(nQ)$ with the same constant depending only $\beta,\mu$. To apply the fix point argument in the next step, we need $\|\kappa'\|_{L^\infty_\om L_x^2(nQ)}$ to be small. This is achieved by requiring  $\|\kappa'\|_{L^\infty_\om L_x^2(Q)}$ to be smaller once $n$ is fixed since $\|\kappa'\|_{L^\infty_\om L_x^2(nQ)}=n^{3/2}\|\kappa'\|_{L^\infty_\om L_x^2(Q)}.$ In conclusion, we may assume from the very beginning that the fundamental domain of $\lat$ has a sufficiently large diameter. This ensures that \eqref{eqn:m-Q-def} is positive.

On the other hand, we claim that an $(n\lat)$ stationary solution associated to a $\lat$ stationary $\kappa$ is in fact $\lat$ stationary, provided such solution in unique. Suppose that $(\varphi(x,\om), \mu) \in L^\infty_\om H^2_x \times \R$ is any $(n\lat)$ stationary unique solution of \eqref{eqn:REHF}. Since $\kappa$ is $\lat$ stationary, for $x\in \R^3$ and $\ell \in \lat$, 
\begin{equation}
     -\frac{1}{4\pi}\Delta \phi(x-\ell,\tau_\ell\om) = \kappa(x, \om) - \den f_{\rm FD}(\beta(-\Delta -\phi(x-\ell,\tau_\ell\om) - \mu)). 
\end{equation}
This shows that $\phi(x-\ell,\tau_\ell\om)$ is also an $(n\lat)$ stationary solution of \eqref{eqn:REHF}. By uniqueness, we see that
\begin{align}
    \phi(x-\ell,\tau_\ell\om)=\phi(x,\om).
\end{align}
Hence, $\phi$ is also $\lat$ stationary. Consequently, it suffices for us to prove Theorem \ref{thm:main-result} and utilize Theorems \ref{thm:LLowerBound-final} under the assumption that $Q$ is sufficiently large. We will make this assumption for the rest of the proof of Theorem \ref{thm:main-result}. At the same time, the uniqueness assumption is established via the fixed point theorem below.

By Theorem \ref{thm:LLowerBound-final}, if $f \in L^\infty_\om L^2_x$, then
\begin{align}\label{eq:mq0}
    \|Lf\|_{L^\infty_\om L^2_x} \gtrsim \|(m_Q-\Delta) f\|_{L^\infty_\om L^2_x},
\end{align}
where $m_Q$ is given in \eqref{eqn:m-Q-def} and 
\begin{align}
    m_Q > 0.
\end{align}
We can assume $m_Q<1/5$ otherwise one can always pick smaller lower bounds so that \eqref{eq:mq0} holds. 
Hence, $L$ is invertible and bounded from below by $m_Q$. Thus, we can write \eqref{eqn:LS-standard-form} as
\begin{equation}
	\phi = L^{-1} \kappa' + L^{-1} N(\phi). \label{eqn:fixedPtPS}
\end{equation}
To apply the fixed point theorem, we work on the ball
\begin{equation}
	B := \{ \phi \in   L^\infty_w H^2_x : \|\phi\|_{ L^\infty_w H^2_x} < \epsilon \} \label{eqn:domain-for-research-for-phi}
\end{equation}
where  $\epsilon \leq \frac{1}{10}C_{\beta,\mu}^{-1}$  is to be determined. If the right hand side of \eqref{eqn:fixedPtPS} were to map $B$ into $B$, then we require
\begin{equation}
	 \|L^{-1} (\kappa' + N(\phi))\|_{L^\infty_w  H^2_x}  \leq \epsilon. \label{eqn:BintoB} 
\end{equation}
 By the choice of $\epsilon \leq \frac{1}{10}C_{\beta,\mu}^{-1}$  and by the nonlinear estimates in Theorem \ref{thm:nonlin}, we see that condition \eqref{eqn:BintoB} is satisfied if
\begin{equation}
	 m_Q^{-1} \left(\| \kappa'\|_{L^\infty_w L^2_x} + C_{\beta,\mu}\|\phi\|_{L^\infty_w H^2_x}^2\, \right)  
	\leq   m_Q^{-1}\, \left (\|\kappa'\|_{L^\infty_w L^2_x} + C_{\beta,\mu}\epsilon^2\, \right)  
	\leq  \epsilon \label{eqn:BBConstraint}
\end{equation}
where $C_{\beta,\mu}$ is given in \eqref{eqn:nonlin}. Moreover, in order that $L^{-1} N$ is a contraction, we require
\begin{equation}
	2m_Q^{-1}\, C_{\beta,\mu} \epsilon < 1 . \label{eqn:contractionConstraint} 
\end{equation}
Consequently, we choose  
\begin{equation}
    \epsilon = \min\left( 2\, m_Q^{-1}\, \|\kappa'\|_{L^\infty_w L^2_x}, \frac{1}{10}C_{\beta,\mu}^{-1} \right) .  
\end{equation}
We see that this choice of $\epsilon$ satisfies \eqref{eqn:contractionConstraint}. If, in addition, $\kappa'$ is sufficiently small:
\begin{equation}
    \|\kappa'\|_{L^\infty_w L^2_x} < \frac{m_Q^2}{4C_{\beta,\mu}}, \label{eq:kappa-small}
\end{equation}
then together with the assumption $m_Q<1/5$, we have
\[2\, m_Q^{-1}\, \|\kappa'\|_{L^\infty_w L^2_x}\le 2\, m_Q^{-1}\,\frac{m_Q^2}{4C_{\beta,\mu}}  \le  \frac{m_Q}{2C_{\beta,\mu}}<\frac{1}{10C_{\beta,\mu}}.\]
Therefore, $ \epsilon =   2\, m_Q^{-1}\, \|\kappa'\|_{L^\infty_w L^2_x}$ and
\[m_Q^{-1}\, \left (\|\kappa'\|_{L^\infty_w L^2_x} + C_{\beta,\mu}\epsilon^2\, \right)=m_Q^{-1}\|\kappa'\|_{L^\infty_w L^2_x}+ (m_Q^{-1} C_{\beta,\mu}\epsilon)\epsilon \le \epsilon/2+\epsilon/2, \]
which verifies \eqref{eqn:BBConstraint}. Theorem \ref{thm:main-result} is now proved by the fixed point theorem, where the solution $\varphi$ of \eqref{eqn:fixedPtPS} is unique on $B$ (see \eqref{eqn:domain-for-research-for-phi}) and satisfies
\begin{equation}
     \|\varphi\|_{L^\infty_w H^2_x}<\eps \leq 2m_Q^{-1} \|\kappa'\|_{L^\infty_w L^2_x} = C_{\beta, \mu, \kappa_0}\|\kappa'\|_{L^\infty_w L^2_x} \label{eq:small-2}
\end{equation}
for some constant $C_{\beta,\mu, \kappa_0}$.  We remark that $C_{\beta,\mu, \kappa_0}$ receives its dependence on $\kappa_0$ through $\mu$ (see Theorem \ref{thm:trans-invar-REHF}).

\end{proof}

\section{Linear analysis} \label{sec:lin-anal}
 The main result of this section is Theorem \ref{thm:LLowerBound-final}, which concerns the lower bound of the operator $L$ on $L^\infty_w L^2_x$ (see \eqref{eqn:linop}). Since this analysis is the core of the paper, we detail its proofs in piece meals. First, we compute an explicit form for $L$ in Subsection \ref{subsec:explicit-form-of-L}. Then we prove a lower bound for $L$ on $L^2(\R^3)$ in Subsection \ref{subsec:lower-on-L2R3}. Finally, building on these two results, we will prove Theorem \ref{thm:LLowerBound-final} in the last Subsection \ref{subsec:pf:thm:LLowerBound-final}.

\subsection{Explicit form of the linearization} \label{subsec:explicit-form-of-L}

We first look for an integral representation of $\den f_{\rm FD}(\beta(-\Delta - \phi -\mu))$ (see the right hand side of \eqref{eqn:REHF}).  Let $\varphi \in H^2(\Omega) \subset L^\infty(\Omega)$ be such that $\|\varphi\|_\infty \lesssim \|\varphi\|_{H^2(\Om)}$ is sufficiently small. Thus, the spectrum of $-\Delta-\varphi$ lays on 
the real axis and bounded from $-\infty$.  Moreover, we note that $f_{\rm FD}(\beta(z-\mu))$ is meromorphic on $\C$ with poles $\mu+i\pi\beta^{-1}(2\Z+1)$. It follows by Cauchy's theorem that
\begin{equation}
    \den f_{\rm FD}(\beta(-\Delta - \phi -\mu)) = \frac{1}{2\pi i} \den \int_{\Gamma} f_{\rm FD}(\beta(z-\mu)) (z-(-\Delta-\phi))^{-1}, \label{eqn:cauchy-int-rep}
\end{equation}
where the contour $\Gamma$ is given in Figure \ref{fig:cauchy-int-contour-general}. In particular, $\Gamma$ is chosen to be at most $\frac{3}{4}\beta^{-1}$ distance away from the real line and  contains the spectrum of $-\Delta-\varphi$.
\begin{figure}[ht]  
\begin{center}
\begin{tikzpicture}
    \path [draw, help lines, opacity=.5]  (-6,-5) grid (6,5);
    
    \foreach \i in {1,...,5} \draw (\i,2.5pt) -- +(0,-5pt) node [anchor=north, font=\small] {} (-\i,2.5pt) -- +(0,-5pt) node [anchor=north, font=\small] {} (-2.92,\i) -- +(-5pt,0) node [anchor=east, font=\small] {} (-2.92,-\i) -- +(-5pt,0) node [anchor=east, font=\small] {};
    
    \draw [->] (-6,0) -- (6,0) node [anchor=south] {$x$};
    \draw [->] (-3,-5) -- (-3,5) node [anchor=west] {$y$};
    
    \draw [<->] (-1,0.1) -- (-1,1.9) node [below=30, right] {$\frac{1}{2}\beta^{-1}$};
    
    \draw [<->] (-4.9,-0.2) -- (-3.1,-0.2) node [below=10, right=-40] {$\frac{1}{2}\beta^{-1}$};

	\draw [-, ultra thick, color=black] (-2, 0) node {} -- (6, 0) node {};
	\node[color=black] at (4,-0.5) {Spectrum};
    
	\draw [->, ultra thick, color=blue, densely dashed] (6,2) node {} -- (-5,2) node {};
    \draw [->, ultra thick, color=blue, densely dashed] (-5,2) node {} -- (-5,-2) node {};
    \draw [->, ultra thick, color=blue, densely dashed] (-5,-2) node {} -- (6,-2) node {};

	\node[color=blue] at (-4, 0.5) {Contour};
  \end{tikzpicture}
\end{center}
\caption{We identify the complex plane $\C$ with $\R^2$ via $z = x+iy$ for $(x,y) \in \R^2$ in the diagram above. The contour $\Gamma$ is denoted by the blue dashed line, extending to positive real infinity. The spectrum of $-\Delta - \varphi$ is contained in the solid black line. \label{fig:cauchy-int-contour-general}}
\end{figure}
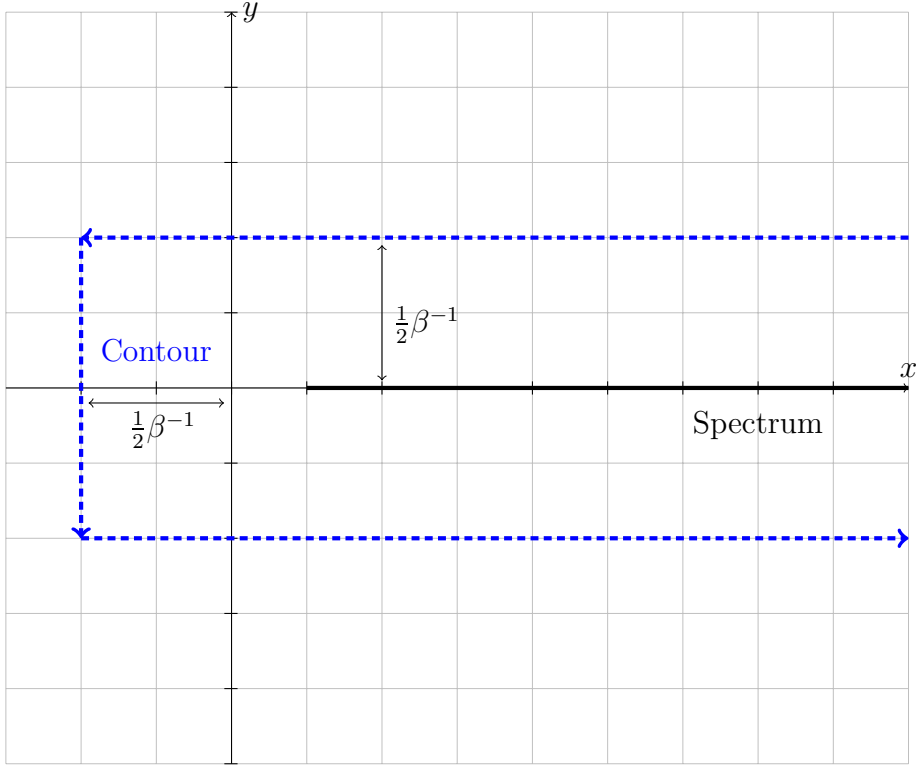
 Since we will be using expressions similar to \eqref{eqn:cauchy-int-rep} repeatedly, we will denote
\begin{equation}
  \oint :=  \frac{1}{2\pi i} \int_\Gamma dz f_{\rm FD}(\beta(z-\mu)) \text{ and }\left| \oint \right| := \frac{1}{2\pi} \int_\Gamma dz |f_{\rm FD}(\beta(z-\mu))|  \label{eqn:oint-def}
\end{equation}
through out the paper.

Now, we are ready to compute the linearization of \eqref{eqn:REHF}. By the resolvent identity
\begin{equation}
	(z-A)^{-1} - (z-B)^{-1} = (z-A)^{-1}(A-B)(z-B)^{-1} \label{eqn:resolventID}
\end{equation}
for any operators $A,B$ on $L^2(\R^3)$, the linear operator (linearized at $\phi = 0$) of the REHF equation \eqref{eqn:REHF} can be seen to be
\begin{align}
	& L = -\frac{1}{4\pi}\Delta + M \text{ where $M$ is defined in \eqref{eqn:linop} and} \notag \\
	&  Mf =  -\den \, \frac{1}{2\pi i}\int_\Gamma dz f_{\rm FD}(\beta(z-\mu))(z+\Delta)^{-1} f (z+\Delta)^{-1}  .  \label{eqn:Mdef}
\end{align}
Notice that we state \eqref{eqn:Mdef} for $f\in L^2_x$, it is also valid for $f\in L_\omega^\infty L^2_x$ with any fixed realization $\omega$. 

 One can verify that $M\ge 0$ is non-negative. The main result of this section is Lemma \ref{lem:L0-explicit-compute} which yields an explicit formula for $M$.

\begin{lemma} \label{lem:L0-explicit-compute}
Assume that $\mu > 0$, then on $L^2(\R^3)$,
 
\begin{equation}
	M = \frac{1}{8\pi^2} \, \frac{1}{ |\nabla|} \int_{0}^\infty \ln \left( \left| \frac{\sqrt{4t} +   |\nabla|}{\sqrt{4t} -  |\nabla|} \right| \right) f_{\rm FD}(\beta(t-\mu)) dt. \label{eqn:MexplicitForm}
\end{equation}
 
\end{lemma}

Before we start the proof for Lemma \ref{lem:L0-explicit-compute}, we state and prove a preliminary estimate for $M$ below. 

\begin{lemma}\label{lem:7}
 $M$ is bounded on $L_w^\infty L_x^2$. 
\end{lemma}
\begin{proof}
First, we make the important remark that the branch cut of any relevant complex function in the rest of the paper is taken to be the negative real axis. For $z \in \C$, 
we consider the associated branch cut of $\sqrt{-z}$, and the integral kernel of $(z-(-\Delta))^{-1}$ is
\begin{equation}\label{eq:gz}
g_z(x-y)=-    \frac{1}{4\pi}\frac{e^{-\sqrt{-z}|x-y|}}{|x-y|}.
\end{equation}
 Then by \eqref{eqn:Mdef}, 
 \begin{equation}
     	Mf =  -  \oint  dz  g_z^2 * f. \label{eqn:MisAConv}
 \end{equation}
We note that  $g_z^2(r) = \frac{1}{16\pi^2}\frac{e^{-2\sqrt{-z}|r|}}{|r|^2}$  is $L^1(\R^3)$ and $f_{\rm FD}(\beta(z-\mu))$ decays exponentially in $\beta \Re (z-\mu)$ as $\Re z \rightarrow \infty$.  It follows by Young's inequality that
\begin{equation}
    \|Mf\|_{L_\om^\infty L^2_x} \lesssim \|f\|_{L_\om^\infty L^2_x}.
\end{equation}

\end{proof}

\begin{proof}[Proof of Lemma \ref{lem:L0-explicit-compute}]
The ideas of this proof is based on an unpublished manuscript of I. Chenn and I. M. Sigal. Let $g_z(x-y)$ be the integral kernel of  $(z-(-\Delta))^{-1}$  as in \eqref{eq:gz}. 
It follows that
\begin{equation}
    (z+\Delta)^{-1} \varphi (z+\Delta)^{-1}(x,x) =   \frac{1}{16\pi^2} \int dy \frac{e^{-2\sqrt{-z}|x-y|}}{|x-y|^2} \varphi(y)  \label{eqn:ugly-unfourier-transformed-lin-op}
\end{equation}
for any sufficiently regular $\varphi$.  To evaluate \eqref{eqn:ugly-unfourier-transformed-lin-op}, we Fourier transform $\frac{e^{-2\sqrt{-z}|x|}}{|x|^2}$: 
\begin{align}
    \int \frac{e^{-2\sqrt{-z}|x|}}{|x|^2} e^{-i x \cdot p } dx 
    &= \int_0^\infty dr \int_0^\pi d\theta \int_0^{2\pi} d\varphi \frac{e^{-2r\sqrt{-z}}}{r^2}  e^{-i \cos(\theta) r|p| }   r^2 \sin(\theta) \\
    &= 4\pi\int_0^\infty dr  \frac{\sin(r|p|)}{r|p|}e^{-2r\sqrt{-z}} \\
    &=  \frac{4\pi\arctan(|p|/\sqrt{-4z})}{|p|}. \label{eqn:arctan-fun-1}
\end{align}
 
By \eqref{eqn:MisAConv}, \eqref{eqn:ugly-unfourier-transformed-lin-op}, and \eqref{eqn:arctan-fun-1},
\begin{equation}
	M= -\frac{1}{8\pi^2} \frac{1}{|\nabla| i }\int_{\Gamma'} f_{\rm FD}(\beta(z-\mu)) \arctan( |\nabla|/\sqrt{-4z}),  \label{eqn:M-1}
\end{equation}
where the contour $\Gamma'$ is homeomorphic to $\Gamma$ (see Figture \ref{fig:cauchy-int-contour-general}) and is defined as follows. Let $0 < \alpha < \beta^{-1}$ be small and $L[a,b]$ denote the line segment connecting $a,b \in \C$. Then, $\Gamma' = L[\infty+ i\alpha, -\alpha +i\alpha] \cup L[-\alpha+ i\alpha, -\alpha -i\alpha] \cup L[-\alpha-i\alpha, \infty-i\alpha]$.

 Next, we simplify \eqref{eqn:M-1} by taking $\alpha \rightarrow 0$. For ease of notation, we will denote both $|\nabla|$ and an arbitrary point in its spectrum by $p$.   We note that
\begin{equation}
	\arctan(z) = \frac{i}{2}\, \Big(\log(1-iz)-\log(1+iz) \Big),
\end{equation}
where $\log$ has the branch cut $(-\infty,0]$.  Moreover, for any $p > 0$ fixed, the integrand 
\[
    f_{\rm FD}(\beta(z-\mu)) \arctan( p/\sqrt{-4z})
\]
has singularities contained in $[0, \infty] \cup (\mu + i\beta^{-1}(2\Z+1))$, otherwise it is holomorphic near the real axis. Let $0 < \alpha$ be small. These observations justify our computation of the integral along another contour which we explicitly parameterize as: 
\begin{align}
	\g_1(t) =& t + \alpha i, t \in (-\alpha, \infty) ,\\
	\g_2(t) =& t - \alpha i, t \in (-\alpha, \infty) ,\\
	\g_3(t) =& -\alpha +  it, t \in (-\alpha, \alpha) \, .
\end{align}
Note that by holomorphicity outside of the singularity of the integrand, the value of the integral is independent of $\alpha$. Hence we may take $\alpha \rightarrow 0$.  For notation, let $f(z) = f_{\rm FD}(\beta(z-\mu))$ and for any path $\gamma : [a,b] \rightarrow \C$, let $\gamma^{-1}$ denote the same path traversed backward. Ignoring pre-factors constants, computing the contour integral \eqref{eqn:M-1} along $\g_1^{-1}$ and $\g_2$, one has 
\begin{align}
	-\frac{1}{8\pi^2}\frac{1}{p i}\int_{-\alpha}^{\infty}  f(t-i\alpha) &\arctan( p/\sqrt{-4(t-i\alpha)}) \notag \\
	&\hspace{2cm}- f(t+i\alpha) \arctan( p/\sqrt{-4(t+i\alpha)}) dt  \notag \\
	=& -\frac{1}{8\pi^2}\frac{1}{p i}\int_{-\alpha}^{\infty} f(t-i\alpha) \arctan( p/\sqrt{-4t+4i\alpha}) \notag \\
	&\hspace{2cm}- f(t+i\alpha) \arctan( p/\sqrt{-4t-4i\alpha}) dt  \label{eqn:contour-int-1}.
\end{align}
Since $f(x)$ is continuous away from its poles and exponentially decaying for  $\Re z > \mu $, in the limit $\alpha \rightarrow 0$, it suffices for us to compute 
\begin{equation}
	\int_{-\alpha}^{\infty} f(t)\, \left (\arctan( p/\sqrt{-4t+4i\alpha}) - \arctan( p/\sqrt{-4t-4i\alpha}) \,  \right)\,  dt \, \label{eqn:contour-int-2}
\end{equation}
in place of \eqref{eqn:contour-int-1}.
We split the integral in \eqref{eqn:contour-int-2} into
\begin{equation}
	\int_{-\alpha}^{\frac{1}{4}p^2} + \int_{\frac{1}{4}p^2}^\infty. \label{eqn:contour-int-3}
\end{equation}
We compute the first integral  in \eqref{eqn:contour-int-3}  and assume that $t < \frac{1}{4}p^2$. We note that \begin{multline}
    \frac{2}{i}\, \left(\arctan(p/\sqrt{-4t+4i\alpha}) - \arctan(p/\sqrt{-4t-4i\alpha}) \, \right)    \\
	=  \log\left( 1 - i \frac{p}{\sqrt{-4t+4i\alpha}} \right) - \log\left( 1 - i \frac{p}{\sqrt{-4t-4i\alpha}} \right)  \\
	 + \log \left( 1+ i \frac{p}{\sqrt{-4t-4i\alpha}} \right) - \log \left( 1+ i \frac{p}{\sqrt{-4t+4i\alpha}} \right)  .\label{eqn:contour-comp-6}
\end{multline}
 Since $t, \alpha > 0$,  to keep track of the jump discontinuity in $\log$, let us write,
\begin{equation}
	 \frac{ip}{\sqrt{-4t \pm 4i\alpha}} =  \frac{ip}{\pm i\sqrt{4t} \sqrt{1 \mp i \frac{\alpha}{t}}}=  \pm  \frac{p}{\sqrt{4t}}  \left(1 \mp i \frac{\alpha}{t}\right)^{-1/2}   \,
\end{equation}
for $t \geq 0$.
Denote $\delta_\pm = \pm  \frac{p}{\sqrt{4t}} ((1 \mp i \frac{\alpha}{t})^{-1/2} - 1)$. Hence
\begin{equation}
	1- \frac{ip}{\sqrt{-4t \pm 4i\alpha}} = 1 \mp \frac{p}{\sqrt{4t}} - \delta_\pm.
\end{equation} 
We note that $\delta_\pm \rightarrow 0$ as $\alpha \rightarrow 0$ and  $\delta_+$ and $\delta_-$  are in the upper half-plane. 
Since $t < \frac{1}{4}  p^2$, using $\lim_{b \rightarrow 0^\pm} \log(-a+ib) = \log(a) \pm i\pi$ for $a > 0$, we see that 
\begin{equation}
	\lim_{\alpha \rightarrow 0} \log \left( 1- \frac{ip}{\sqrt{-4t + 4i\alpha}} \right) = \log\left( \frac{p}{\sqrt{4t}} - 1 \right) - i\pi 
\end{equation}
and similarly
\begin{equation}
	\lim_{\alpha \rightarrow 0} \log \left( 1- \frac{ip}{\sqrt{-4t - 4i\alpha}} \right) = \log\left( \frac{p}{\sqrt{4t}} + 1 \right).
\end{equation}
So we have that
\begin{multline}
   	\lim_{\alpha \rightarrow 0}   \log \left( 1- \frac{ip}{\sqrt{-4t + 4i\alpha}} \right) - \log \left( 1- \frac{ip}{\sqrt{-4t - 4i\alpha}} \right) \\
	=  \log\left( \frac{p}{\sqrt{4t}} - 1 \right)  -  \log\left( \frac{p}{\sqrt{4t}} + 1 \right)-i\pi . \label{eqn:contour-comp-5} 
\end{multline}

By the same token, and  since $t < \frac{1}{4}p^2$,
\begin{align}
	1 + \frac{ip}{\sqrt{-4t \pm 4i\alpha}} = 1 \pm \frac{p}{\sqrt{4t}} + \delta_\pm.
\end{align}
Using $\lim_{b \rightarrow 0^\pm} \log(-a+ib) = \log(a) \pm i\pi$ for $a > 0$, we see that
\begin{align}
	\lim_{\alpha \rightarrow 0} \log \left( 1 + \frac{ip}{\sqrt{-4t - 4i\alpha}} \right) = \log\left( \frac{p}{\sqrt{4t}} - 1 \right) + i\pi ,
\end{align}
and similarly
\begin{align}
	\lim_{\alpha \rightarrow 0} \log \left( 1 + \frac{ip}{\sqrt{-4t + 4i\alpha}} \right) = \log\left( \frac{p}{\sqrt{4t}} +1 \right).
\end{align}
So we have that
\begin{multline}
    \lim_{\alpha \rightarrow 0}    \log \left( 1 + \frac{ip}{\sqrt{-4t - 4i\alpha}} \right) - \log \left( 1 + \frac{ip}{\sqrt{-4t + 4i\alpha}} \right) \\
		=  \log\left( \frac{p}{\sqrt{4t}} - 1 \right) - \log\left( \frac{p}{\sqrt{4t}} +1 \right)  +i\pi  . \label{eqn:contour-comp-4}
\end{multline}
It follows  by \eqref{eqn:contour-comp-6}, \eqref{eqn:contour-comp-5}, \eqref{eqn:contour-comp-4}  that
\begin{multline}
    	\lim_{\alpha \rightarrow 0}    \arctan(  p/\sqrt{-4t+4i\alpha}) - \arctan(  p/\sqrt{-4t-4i\alpha})   \\
		=  i \left( \log\Big( \frac{p}{\sqrt{4t}} - 1 \Big) - \log\Big( \frac{p}{\sqrt{4t}} +1 \Big) \right). \label{eqn:contour-comp-7}
\end{multline}
Hence, in the limit $\alpha \rightarrow 0$, the $\g_1^{-1},\g_2$ portion of the  first integral in \eqref{eqn:contour-int-3}  is
\begin{multline}
 \int_0^{p^2/4}= i\int_0^{ p^2/4} \left( \log\Big( \frac{p}{\sqrt{4t}} - 1 \Big) - \log\Big( \frac{p}{\sqrt{4t}} +1 \Big) \right) f(t) dt \\
	= -i\int_0^{p^2/4}  \log\left( \frac{p + \sqrt{4t}}{p-\sqrt{4t}} \right) f(t) dt.
\end{multline}

Now we consider the second integral  in \eqref{eqn:contour-int-3}.
In this case, for $t >  p^2/4 $, 
\begin{multline}
    \frac{2}{i}\, \left(\, \arctan(  p/\sqrt{-4t+4i\alpha}) - \arctan(  p/\sqrt{-4t-4i\alpha}) \, \right)\\
	=  \log\left( 1 - i \frac{p}{\sqrt{-4t+4i\alpha}} \right) - \log\left( 1 - i \frac{p}{\sqrt{-4t-4i\alpha}} \right)   \\
	 + \log \left( 1+ i \frac{p}{\sqrt{-4t-4i\alpha}} \right) - \log \left( 1+ i \frac{p}{\sqrt{-4t+4i\alpha}} \right). 
\end{multline}
 Via the same argument as \eqref{eqn:contour-comp-7},  we get
\begin{multline}
  \arctan(  p/\sqrt{-4t+4i\alpha}) - \arctan(  p/\sqrt{-4t-4i\alpha}) \\
	\longrightarrow
	i \Big(\log \left( 1-  \frac{p}{\sqrt{4t}} \right) -   \log\left( 1 + \frac{p}{\sqrt{4t}} \right)  \Big)
	= i\log \left( \frac{\sqrt{4t}-  p}{\sqrt{4t}+p} \right) 
\end{multline}

as $\alpha \rightarrow 0$.  Putting it as the integrand for the second integral in \eqref{eqn:contour-int-3}, we get
\begin{equation}
    \int_{ p^2/4}^\infty  =  i \int_{p^2/4}^\infty  \log \left( \frac{\sqrt{4t}-   p}{\sqrt{4t}+ p} \right)f(t) dt 
		 =   -i\int_{p^2/4}^\infty \log \left( \frac{\sqrt{4t}+  p}{\sqrt{4t}- p} \right)f(t) dt.
\end{equation}

It follows that by \eqref{eqn:M-1} and \eqref{eqn:contour-int-1}
\begin{multline}
	M= -\frac{1}{8\pi^2} \frac{1}{p i }  \int_0^{ p^2/4}      f(t)\, \left (\arctan( p/\sqrt{-4t+4i\alpha}) - \arctan( p/\sqrt{-4t-4i\alpha}) \,  \right)\,  dt \, \\
	-\frac{1}{8\pi^2} \frac{1}{p i } \int_{ p^2/4}^\infty f(t)\, \left (\arctan( p/\sqrt{-4t+4i\alpha}) - \arctan( p/\sqrt{-4t-4i\alpha}) \,  \right)\,  dt \\
=-\frac{1}{8\pi^2} \frac{1}{p i } \Big[	-i\int_0^{p^2/4}  \log\left( \frac{p + \sqrt{4t}}{p-\sqrt{4t}} \right) f(t) dt\Big]\\
	-\frac{1}{8\pi^2} \frac{1}{p i }\Big[-i\int_{p^2/4}^\infty \log \left( \frac{\sqrt{4t}+  p}{\sqrt{4t}- p} \right)f(t) dt\Big] \\
=\frac{1}{8\pi^2} \frac{1}{p  }\int_{0}^\infty \log \left( \left|\frac{\sqrt{4t}+  p}{\sqrt{4t}- p}\right| \right)f(t) dt,
\end{multline}
which is \eqref{eqn:MexplicitForm}.
 The proof of Lemma \ref{lem:L0-explicit-compute} is now complete. 
\end{proof}

\subsection{Lower bound of $L$ on $L^2(\R^3)$} \label{subsec:lower-on-L2R3}
\begin{theorem} \label{thm:L-lower-on-R3}
The operator $L$ given in \eqref{eqn:Mdef} is bounded below on $L^2(\R^3)$ by
\begin{equation}\label{eq:L-lower1}
    L \ge c_0( -\Delta + m_*)
\end{equation}
where $c_0>0$ is a universal constant and 
\begin{equation}
    m_* = \min(\mu,\sqrt{\mu}). \label{eqn:m-star-def}
\end{equation}
\end{theorem}
\begin{proof}
 Let $p = |\nabla|$.  By elementary calculus, we have the following two estimates
\begin{align}
	\ln\left| \frac{x+1}{x-1} \right| > 2x \text{ for } x < 1 \\
	\ln\left| \frac{x+1}{x-1} \right| > \frac{2}{x} \text{ for } x > 1. 
\end{align}
 In the eigenspace with $p^2/4 < \mu$,  we use $f_{\rm FD}(\beta(t-\mu)) \geq \frac{1}{2}\chi_{[0,\mu]}$, where $\chi_S$ is the indicator function on $S \subset \R$. For a particular $p^2/4 < \mu$ fixed, equation \eqref{eqn:MexplicitForm} shows
\begin{align}
	M \gtrsim &  \int_0^{p^2/4} \frac{\sqrt{t}}{p^2} +  \int_{p^2/4}^\mu \frac{1}{\sqrt{t}} \\
		= & \left. \frac{1}{p^2} \frac{2}{3} t^{3/2} \right|_0^{p^2/4} + \left. 2 t^{1/2} \right|_{p^2/4}^\mu \\
		\geq & \frac{1}{6}\sqrt{\mu}. \label{eqn:L-lower-sqrt-mu}
\end{align}
On the other hand, in the eigenspace with $p^2/4 \ge  \mu$, we use $L=\frac{1}{8\pi}p^2+ (\frac{1}{8\pi}p^2+M)\ge \frac{1}{8\pi}p^2+ ( \frac{1}{8\pi}p^2)$. We see that $\frac{1}{8\pi}p^2  \gtrsim  \mu$.  Combining with \eqref{eqn:L-lower-sqrt-mu}, Theorem \ref{thm:L-lower-on-R3} is proved. 
\end{proof}

\subsection{Lower bound for $L$ on $L^\infty_w L^2_x$: proof of Theorem \ref{thm:LLowerBound-final}} \label{subsec:pf:thm:LLowerBound-final}
By Lemma \ref{lem:7}, $L$ is bounded operator on $L^\infty_w L^2_x$. In this subsection, we want to extend the lower bound \eqref{eq:L-lower1} also to $L^\infty_w L^2_x$. 
From this subsection on, we will deal with two types of Laplacians: the usual Laplacian $-\Delta$ on $L^2_x(\R^3)$, and the  stationary Laplacian $-\Delta_s$ on $L_\omega^2 L^2_x$. The stationary  Laplacian is simply the   usual Laplacian in  $x$ variables,  acting  on $L^2(\Omega\times Q)$,  with  stationary  boundary  conditions on $Q\subset \R^3$.  The properties of the  stationary Laplacian $-\Delta_s$ have been studied in e.g., \cite{DS}, see also in Section 3.1 of \cite{CLL}. We may frequently drop the lower index $s$ and denote by $-\Delta_s=-\Delta$ when there is no ambiguity, especially when the random realization $\omega$ is fixed (almost surely). 

\begin{proof}[Proof of Theorem \ref{thm:LLowerBound-final}]
Let $Q$ be a fundamental domain of $\lat$ and $f \in L^\infty_\om L^2_x$ and $C_c^2(Q)$ denote the set of compactly supported $C^2$ function on $Q$. Then
\begin{align}
	\|(-\frac{1}{4\pi}\Delta_s + M)f\|_{L^\infty_\om L^2_x} :=& \sup_{\om \in \Om} \|(-\frac{1}{4\pi}\Delta_s + M)f\|_{L^2(Q)} \\
		=& \sup_{\om \in \Om}\sup_{\psi \in C_c^2(Q)} \frac{\lan \psi, (-\frac{1}{4\pi}\Delta_s + M)f \ran_{L^2(Q)}}{\|\psi\|_{L^2(Q)}}  . \label{eqn:start1}
\end{align}
The goal is to choose a nice $\psi$ for a lower bound. Intuitively, we should choose $\psi = (-\frac{1}{4\pi}\Delta_s + M)^{-1} f$. We accomplish this in 5 steps below.

For notation simplicity below, we will first fix $\omega$ and follow \eqref{eqn:oint-def} and also use
\begin{equation}
	r(z)  := (z+\Delta)^{-1}.
\end{equation}
In this notation, the linear operator is
\[
	Lf = -\frac{1}{4\pi}\Delta f + Mf = -\frac{1}{4\pi}\Delta f - \den \oint r(z) f r(z). 
\]

\textbf{Step 1: Move $-\frac{1}{4\pi}\Delta + M$ onto $\psi$.} \\
We have a solid understanding of $-\frac{1}{4\pi}\Delta + M$ on $L^2(\R^3)$ functions; so we would like to use it.
Since $\psi \in C_c^2(Q)$, we see that
\begin{equation}
	\lan \psi, (-\frac{1}{4\pi}\Delta) f \ran_{L^2(Q)} = \lan (-\frac{1}{4\pi}\Delta)\psi, f \ran_{L^2(Q)}. \label{eqn:DeltaSymmetric}
\end{equation}
Using the definition \eqref{eqn:Mdef} of $M$, we see that
\begin{equation}
	\lan \psi, M f \ran_{L^2(Q)} = -\oint \Tr \bar \psi r(z) f r(z) = -\oint \Tr \overline{r(\bar z) \psi r(\bar z)} f = \lan M \psi, f \ran_{L^2(\R^3)} \label{eqn:MSymmetric}
\end{equation}
by the cyclicity of trace and the fact the contour $\Gamma$ can be chosen to be invariant under complex conjugation. We remark that since   $\psi \in C^2_c(Q)$, $\Tr$ is taken to be the trace on $L^2(\R^3)$ instead of $L^2(Q)$. By equations \eqref{eqn:DeltaSymmetric} and \eqref{eqn:MSymmetric}, we see that
\begin{equation}
	\lan \psi, (-\frac{1}{4\pi}\Delta + M)f \ran_{L^2(Q)} = \lan (-\frac{1}{4\pi}\Delta + M)\psi, f \ran_{L^2(\R^3)} \label{eqn:DMSymmetric}.
\end{equation}

\textbf{Step 2: choosing $\psi$.} \\
By Theorem \ref{thm:L-lower-on-R3}, $-\frac{1}{4\pi}\Delta + M$ is bounded from below on $L^2(\R^3)$. One has 
\[\|(-\frac{1}{4\pi}\Delta + M)g\|_{L^2(\R^3)}\ge c_0 \|(-\Delta +m_\ast)g\|_{L^2(\R^3)}\ge c_0 m_\ast\| g\|_{H^2(\R^3)}.  \] In particular, $-\frac{1}{4\pi}\Delta + M$ is invertible and 
\begin{equation} \label{eq:L-1}
  \|(-\frac{1}{4\pi}\Delta + M)^{-1}(f\mid_Q)\|_{H^2(\R^3)}\le c_0^{-1} m_\ast^{-1} \| f\mid_Q \|_{L^2(\R^3)} .  
\end{equation}
Thus, we choose
\begin{equation}
	\psi :=  \chi (-\frac{1}{4\pi}\Delta + M)^{-1}(f\mid_Q), \label{eqn:psiDef}
\end{equation}
where $\chi$ is a compactly supported bump function on $Q$ whose exact properties will be determined later and $f\mid_Q$ is the restriction of $f$ to $Q$. 
We remark that despite the fact that $\psi$ may not be in $C_c^2(Q)$, it can be approximated arbitrarily close in the $H^2$ norm by functions from $C_c^2(Q)$. Thus, it causes no disruption to our proof to consider $\psi$ instead of a proper $C_c^2(Q)$ function. Moreover, this choice of $\psi$ satisfies \eqref{eqn:DeltaSymmetric}.  It follows that
\begin{equation}
	(-\frac{1}{4\pi}\Delta + M)\psi = \chi f\mid_Q + [-\frac{1}{4\pi}\Delta + M, \chi](-\frac{1}{4\pi}\Delta + M)^{-1}(f\mid_Q).
\end{equation}
Combining with \eqref{eqn:DMSymmetric}, we see that
\begin{multline}
   	\lan \psi, (-\frac{1}{4\pi}\Delta + M)f \ran_{L^2(Q)} \\ = \|\chi^{1/2} f\|^2_{L^2(Q)} + \lan [-\frac{1}{4\pi}\Delta + M, \chi](-\frac{1}{4\pi}\Delta + M)^{-1}(f\mid_Q), f \ran_{L^2(\R^3)}. \label{eqn:main1} 
\end{multline}
We estimate the  second term in \eqref{eqn:main1}  below.

\textbf{Step 3: Error estimates for $[\Delta, \chi]$ term.} \\
A simple computation shows that
\begin{equation}
	[-\Delta, \chi] = -2\nabla \chi \cdot \nabla - \Delta \chi.
\end{equation}
Using Theorem \ref{thm:L-lower-on-R3} once more and the fact $\chi$ is supported in $Q$, it follows that
\begin{multline}
   |\lan [-\Delta, \chi](-\frac{1}{4\pi}\Delta + M)^{-1}(f\mid_Q), f \ran_{L^2(\R^3)}|  
   = |\lan [-\Delta, \chi](-\frac{1}{4\pi}\Delta + M)^{-1}(f\mid_Q), f \ran_{L^2(Q)}| \\ 
 \le\Big| \lan 2\nabla \chi \cdot \nabla(-\frac{1}{4\pi}\Delta + M)^{-1}(f\mid_Q), f \ran_{L^2(Q)} \Big|  
 +\Big| \lan \Delta \chi(-\frac{1}{4\pi}\Delta + M)^{-1}(f\mid_Q), f \ran_{L^2(Q)} \Big| \\
 \le2 \|\nabla \chi\|_{L^\infty}\|\nabla (-\frac{1}{4\pi}\Delta + M)^{-1}(f\mid_Q)\|_{L^2(Q)}\| f \|_{L^2(Q)} \\
 +\|\Delta \chi\|_{L^\infty} \| (-\frac{1}{4\pi}\Delta + M)^{-1}(f\mid_Q)\|_{L^2(Q)}\| f \|_{L^2(Q)} \\
 \le  2 \|\nabla \chi\|_{L^\infty}\, c_0^{-1}m_*^{-1} \| f \|_{L^2(Q)}^2 
 +\|\Delta \chi\|_{L^\infty}  \, c_0^{-1}m_*^{-1} \| f \|_{L^2(Q)}^2 \\ 
 = c_0^{-1}m_*^{-1}\Big(2 \|\nabla \chi\|_{L^\infty} 
 +\|\Delta \chi\|_{L^\infty} \Big ) \| f \|_{L^2(Q)}^2 \label{eqn:main2}
\end{multline}
where $c_0,m_*$ are given in \eqref{eqn:m-star-def} and we used \eqref{eq:L-1}.

\textbf{Step 4: Error estimates for $[M, \chi]$ term.} \\
Let $u \in L^2(\R^3)$ and $v \in L^2_{\rm loc}(\R^3)$. Then
\begin{align}
	\lan [M, \chi]u, v \ran_{L^2(\R^3)} =& \oint \Tr r(z) \chi u r(z) v  - \chi  r(z) u r(z) v \\
		=& \oint \Tr [r(z), \chi]ur(z) v \\
		=& \oint \Tr r(z)(2\nabla \chi \nabla + \Delta \chi) r(z) u r(z) v .
\end{align}
By the cyclicity of trace, we move $r(z)$ to the right of $v$ to obtain
\begin{equation}
	\lan [M, \chi]u, v \ran_{L^2(\R^3)}	=  \oint \Tr (2\nabla \chi \nabla + \Delta \chi) r(z) u r(z) v r(z). 
\end{equation}
Since $\chi$ is supported on $Q$, we may write  
\begin{align}
	\lan [M, \chi]u, v \ran_{L^2(\R^3)}	=  \oint \Tr {\bf 1}_Q (2\nabla \chi \nabla + \Delta \chi)  r(z) u r(z) v r(z){\bf 1}_Q, \label{eqn:some-eqn-1}
\end{align}
where ${\bf 1}_Q$ is the indicator function of the set $Q$.  Let $g_z$ be the integral kernel of $r(z) =(z+\Delta)^{-1}$ on $L^2(\R^3)$ and recall the definition of $\left|\oint\right|$ from \eqref{eqn:oint-def}. Writing \eqref{eqn:some-eqn-1} out explicitly in terms of integral kernels, we obtain
\begin{multline}
   |\lan [M, \chi]u, v \ran_{L^2(\R^3)}|
	\leq   2(\|\nabla \chi\|_{L^\infty} + \|\Delta \chi\|_{L^\infty})  \\
  \times \left|\oint\right|  \int_Q dx \int_{\R^{3+3}} dy_1 dy_2   
	 |(1+\nabla) g_z(x-y_1)u(y_1) g_z(y_1-y_2) v(y_2) g_z(y_2-x)|. \label{eqn:Mchi1} 
\end{multline}
Using a change of variable $y_1 \mapsto y_1 + x$ and $y_2 \mapsto y_2 + x$, we see that
\begin{multline}
     \left|\oint\right|  \int_Q dx \int_{\R^{3+3}} dy_1 dy_2  
 |(1+\nabla)g_z(x-y_1)u(y_1) g_z(y_1-y_2) v(y_2) g_z(y_2-x)|  \\
	=  \left|\oint\right|  \int_Q dx \int_{\R^{3+3}} dy_1 dy_2  |(1+\nabla)g_z(y_1)u(y_1+x) g_z(y_1-y_2) v(y_2+x) g_z(y_2)|. \label{eqn:Mchi2}
\end{multline}
Applying H\"{o}lder with $\frac{1}{2} + \frac{1}{2} = 1$ to the $dx$ integral, we see that
\begin{multline}
    \left|\oint\right|  \int_Q dx \int_{\R^{3+3}} dy_1 dy_2   |(1+\nabla)g_z(y_1)u(y_1+x) g_z(y_1-y_2) v(y_2+x) g_z(y_2)|  \\
	\leq   \left|\oint\right| \int_{\R^{3+3}} dy_1 dy_2 |(1+\nabla)g_z(y_1)g_z(y_1-y_2)g_z(y_2)| \|u\|_{L^2(Q+y_1)} \|v\|_{L^2(Q+y_2)} \\
	\leq   \left|\oint\right| \int_{\R^{3+3}} dy_1 dy_2 |(1+\nabla)g_z(y_1)g_z(y_1-y_2)g_z(y_2)| \|u\|_{L^2(\R^3)} \|v\|_{L^2(Q+y_2)}. \label{eqn:Mchi3}
\end{multline}
We specialize to $u = (-\frac{1}{4\pi}\Delta + M)^{-1}(f\mid_Q)$ and $v = f$. By Theorem \ref{thm:L-lower-on-R3}, we note that
\begin{equation}
	\|u\|_{L^2(\R^3)} \le c_0^{-1} m_*^{-1} \|f\|_{L^2(Q)}, \label{eqn:Mchi4}
\end{equation}
where $c_0,m_*$ are given in \eqref{eqn:m-star-def}.
Combining \eqref{eqn:Mchi1}, \eqref{eqn:Mchi2}, \eqref{eqn:Mchi3}, and \eqref{eqn:Mchi4}, we see that
\begin{align}
	& \sup_{\om \in \Om} |\lan [M, \chi]u, v \ran_{L^2(\R^3)}|	\notag \\
& \quad 	\le   2c_0^{-1}m_*^{-1} (\|\nabla \chi\|_{L^\infty} + \|\Delta \chi\|_{L^\infty})  \|f\|_{L^\infty_\om L^2_x}\notag \\
	& \quad \quad \times \left|\oint \right| \int_{\R^{3+3}} dy_1 dy_2 |(1+\nabla)g_z(y_1)g_z(y_1-y_2)g_z(y_2)|  \sup_{\om \in \Om} \|f\|_{L^2(Q+y_2)}.
\end{align}
 
Let
\begin{equation}
    A = 2 c_0^{-1}(\|\nabla \chi\|_{L^\infty} + \|\Delta \chi\|_{L^\infty}). 
\end{equation}
Recall the definition of $g_z$ in \eqref{eq:gz}. It is easy to verify that $g_z\in W^{1,1}(\R^3)$, and $\| |g_z| * |g_z|\|_{L^\infty}\lesssim \|g_z\|^2_{L^2}$. Since the $L^\infty_\om L^2_x$ norm is independent  of the location  of the domain of integration, we see that
\begin{align}
	 \sup_{\om \in \Om} |\lan [M, \chi]u, v \ran_{L^2(\R^3)}| 
	\leq & m_*^{-1}A \|f\|_{L^\infty_\om L^2_x}^2 \left|\oint \right| \lan |(1+\nabla)g_z|, |g_z| * |g_z| \ran_{L^2(\R^3)} \\
	\leq & m_*^{-1}A \|f\|_{L^\infty_\om L^2_x}^2 \left|\oint \right| \Big\| |g_z| * |g_z|\Big\|_{L^\infty} \|(1+\nabla)g_z\|_{L^1(\R^3)} \\
	\leq & m_*^{-1}A \|f\|_{L^\infty_\om L^2_x}^2 \left|\oint \right| \Big\| \|g_z\|^2_{L^2(\R^3)} \|g_z\|_{W^{1,1}(\R^3)} \\ 
	=: & m_*^{-1}A C_{\beta, \mu}\|f\|_{L^\infty_\om L^2_x}^2,  \label{eqn:main3}
\end{align}
where
\begin{equation} \label{eq:Cbetamu}
    C_{\beta, \mu} = \left|\oint \right| \Big\| \|g_z\|^2_{L^2(\R^3)} \|g_z\|_{W^{1,1}(\R^3)} < \infty.
\end{equation}

\textbf{Step 5: Conclusion} \\
By definition \eqref{eqn:psiDef} and \eqref{eq:L-1}, we see that
\begin{equation}
	\|\psi\|_{L^2(Q)} \leq c_0^{-1} m_*^{-1} \|f\|_{L^2(Q)}.
\end{equation}
Combining with equations \eqref{eqn:main1}, \eqref{eqn:main2}, and \eqref{eqn:main3} to obtain
 \begin{equation}
    \sup_{\om \in \Om} \frac{\lan \psi, (-\frac{1}{4\pi}\Delta_s + M)f \ran_{L^2(Q)}}{\|\psi\|_{L^2(Q)}}   \\
	\geq   c_0 m_* \sup_{\om \in \Om} \frac{\|\chi^{1/2} f\|^2_{L^2(Q)}}{\|f\|_{L^2(Q)}} - 2C_{\beta, \mu}  A\|f\|_{L^\infty_\om L^2_x} . \label{eqn:main4}
\end{equation} 
 
Choose $\om_0 \in \Om$ so that
\begin{equation}
	\|f_{\om_0}\|_{L^2(Q)}^2 \geq \frac{1}{2} \sup_{\om \in \Om} \|f_\om\|_{L^2(Q)}^2. \label{eqn:QQp1}
\end{equation}
 Assuming $Q$ is a cube centered at the origin (or by a shift and choosing a different fundamental domain if necessary), we can find a subcube $Q' \subset Q$ with half the diameter, such that  
\begin{equation}
	\|f_{\om_0}\|_{L^2(Q')}^2 \geq \frac{1}{2^3}\|f_{\om_0}\|_{L^2(Q)}^2.
\end{equation} 
We can translation $Q'$ to be concentric to $Q$ by $h \in \R^3$ so that  
\begin{equation}
	\|T_h f_{\om_0}\|_{L^2(Q/2)}^2 \geq \frac{1}{2^3}\|f_{\om_0}\|_{L^2(Q)}^2, \label{eqn:QQp2}
\end{equation} 
where $T_h$ is the translation by $h$ and $Q/2$ is the concentric cube to $Q$ with half the side length. 
Thus, we can choose bump functions $\chi$ by a translation and a dilation from a reference $O(1)$ bounded $C_c^\infty$ function so that $\chi \mid_{Q/2} = 1$ and
\begin{equation}
	\|\nabla^m \chi\|_{L^\infty} \leq C_1 \ell(Q)^{-m}, \label{eqn:nablachi}
\end{equation}
where $C_1$ is independent of the size of $Q$ and $\ell(Q)$ is the diameter of $Q$. Since the operator $-\frac{1}{4\pi}\Delta_s + M$ is translation invariant, and the $L^\infty_w L^2_x$ norm is independent of the choice of $Q$, we may replace $f$ by any translations of $f$. In doing so to \eqref{eqn:main4} and using \eqref{eqn:QQp1}, \eqref{eqn:QQp2}, and the fact $\chi\mid_{Q/2} = 1$, we obtain
\begin{align}
c_0 m_* \sup_{\om \in \Om} \frac{\|\chi^{1/2} T_h f\|^2_{L^2(Q)}}{\| T_h f\|_{L^2(Q)}} -&  2A C_{\beta, \mu} \| T_h f\|_{L^\infty_\om L^2_x} \notag \\
	\geq & c_0m_* \sup_{\om \in \Om} \frac{\|T_h f\|^2_{L^2(Q/2)}}{\| T_h f\|_{L^2(Q)}} - 2A C_{\beta, \mu} \|f\|_{L^\infty_\om L^2_x}\notag \\
	\geq & c_0m_* \frac{\|T_h f_{\om_0}\|^2_{L^2(Q/2)}}{\|f\|_{L^\infty_\om L^2_x}} - 2A C_{\beta, \mu} \|f\|_{L^\infty_\om L^2_x} \notag \\
	\geq & \left(\frac{c_0m_*}{2^4}  - 2A C_{\beta, \mu} \right)\,  \|f\|_{L^\infty_\om L^2_x}  . \label{eqn:final}
\end{align}
It follows by translation invariance of $-\frac{1}{4\pi}\Delta_s + M$ and $Q$-independence of the $L^\infty_\om L^2_x$ norm, \eqref{eqn:start1}, \eqref{eqn:main4}, \eqref{eqn:nablachi}, and \eqref{eqn:final},
\begin{align}
	\|(-\frac{1}{4\pi}\Delta_s + M)f\|_{L^\infty_\om L^2_x} =& \|(-\frac{1}{4\pi}\Delta_s + M)T_hf\|_{L^\infty_\om L^2_x} \notag \\
	\geq& \sup_{\om \in \Om} \frac{\lan \psi, (-\frac{1}{4\pi}\Delta_s + M) T_hf \ran_{L^2(Q)}}{\|\psi\|_{L^2(Q)}} \notag \\
	\geq & \left(\frac{c_0m_*}{16}  - C_{\beta,\mu}   {\ell(Q)^{-1}} \right) \|f\|_{L^\infty_\om L^2_x}  , \label{eqn:main5}
\end{align}
where $C_{\beta,\mu}$ is a constant multiple of the one defined in \eqref{eq:Cbetamu}. 
This proves Theorem \ref{thm:LLowerBound-final}.
\end{proof}

\section{Nonlinear estimate} \label{sec:nonlin-anal}

\begin{proof}[Proof of Theorem \ref{thm:nonlin}]
 By the resolvent identity in \eqref{eqn:resolventID} and using the Cauchy-integral, we arrive at an explicit formula for $N$ (see \eqref{eqn:nonlinDef}):
 
\begin{equation}
	N(\phi) :=  \den \oint (z-(-\Delta - \phi))^{-1} \phi (z+\Delta)^{-1}\ +M\phi, \label{eqn:nonlinExplicit}
\end{equation} 
where $\oint$ is defined through \eqref{eqn:oint-def}.

Applying the resolvent identity \eqref{eqn:resolventID} to \eqref{eqn:nonlinExplicit} repeatedly with $A = -\Delta - \phi$ and $B = -\Delta$, we arrive at
\begin{equation}
	N(\phi) = \sum_{n \geq 2}  (-1)^{n-1}\oint\den (z+\Delta)^{-1} (\phi (z+\Delta)^{-1})^n , \label{eqn:NexpandinNn}
\end{equation}
whenever the series converges in $L^\infty_\om L^2_x$. Our goal is to estimate the difference in the individual $n$-th order nonlinearities. That is, let
\begin{align}
     N_n(\phi) := (-1)^{n-1}\oint \den (z+\Delta)^{-1} (\phi (z+\Delta)^{-1})^n \label{eqn:N-n-def}.
\end{align}
We would like to estimate $N_n(\phi_1) - N_n(\phi_2)$.  To do so, we use 
\begin{align}
    a^b - b^n = (a-b)b^{n-1} + a(a-b)b^{n-2} + \cdots + a^{n-1}(a-b).    
\end{align} 
Applying this to the $n$-fold products of resolvents in $N_n(\phi_1)$ and $N_n(\phi_2)$, and writing out the first term explicitly, we see that
\begin{align}
 N_n(\phi_1)& - N_n(\phi_2)	 \notag \\
	=& (-1)^{n-1}\oint \den (z+\Delta)^{-1}  \Big( (\phi_1 (z+\Delta)^{-1})^n - (\phi_2 (z+\Delta)^{-1})^n \Big)  \label{eqn:NnDef} \\
	=& (-1)^{n-1}\oint \den  {(z+\Delta)^{-1} (\phi_1 - \phi_2)(z+\Delta)^{-1}}(\phi_\# (z+\Delta)^{-1})^{n-1} \\
	&+ \text{$n-1$ similar terms}
\end{align}
in the $L^\infty_\om L^2_x$-norm, where $\phi_\#$ denotes $\phi_1$ or $\phi_2$.
 Let $g_z(x-y)$ be the integral kernel of $(z+\Delta)^{-1}$ as in \eqref{eq:gz}.  
Then the $n$-th order (difference in) nonlinearity \eqref{eqn:NnDef} becomes
\begin{align}
	 \big(N_n(\phi_1)& - N_n(\phi_2)\big)(x) \notag \\
	=& (-1)^{n-1}\oint \int dy_1 \cdots dy_n g_z(y_1)g_z(y_1-y_2) \cdots g_z(y_{n-1}-y_n)g_z(y_n) \notag \\
		&\times (\phi_1-\phi_2)(x-y_1)  {\prod_{i=2}^{n}} \phi_\#(x-y_i)  \notag \\
		&\quad +\text{$n-1$ similar terms}. \notag
\end{align}
We may take the sup-norm on the factor $\prod_{i=2}^{n-1} \phi_\#(x-y_i)$.  {Using the notation $\left|\oint\right|$ in \eqref{eqn:oint-def}, we see that}
\begin{align}
	 \Big| \big(N_n(\phi_1)& - N_n(\phi_2)\big)(x) \Big| \notag \\
	\leq & \|\phi_\#\|_{L_\omega^\infty L^\infty_x(Q)}^{n-1} \left|\oint\right|  \int dy_1 |g_z|(y_1)|\phi_1-\phi_2|(x-y_1) \underbrace{(|g_z| * \cdots * |g_z|)}_{n}(y_1)   \notag\\
		&\quad +\text{$n-1$ similar terms}  \notag \\
	=& \|\phi_\#\|_{L_\omega^\infty L^\infty_x(Q)}^{n-1} \left|\oint\right|  k_z *|\phi_1-\phi_2|\quad +\text{$n-1$ similar terms}   \label{eqn:Ndiff1}\\
	=& \|\phi_\#\|_{L_\omega^\infty L^\infty_x(Q)}^{n-1} \left|\oint\right|  k_z *|\phi_1-\phi_2|\quad +\text{$n-1$ similar terms},  
\end{align}
where $k_z(x) := |g_z|(x) (|g_z| * \cdots * |g_z|)(x)$. It follows by \eqref{eqn:Ndiff1} that
\begin{equation}
	\| N_n(\phi_1) - N_n(\phi_2)\|_{L^2(Q)} \leq n\|\phi_\#\|_{L_\omega^\infty L^\infty_x(Q)}^{n-1} \left|\oint\right| \|k_z*|\phi_1-\phi_2|\|_{L^2(Q)}. \label{eqn:Ndiff2}
\end{equation}
By H\"{o}lder's inequality, we note that
\begin{align}
    	\|k_z*|\phi_1-\phi_2|\|_{L^2(Q)}^2 =& \int_Q dx \left( \int_{\R^3} dy k_z(x-y) |\phi_1 - \phi_2|(y) \right)^2  \notag\\
		\leq & \int_Q dx \|k_z\|_{L^1(\R^3)} \int_{\R^3} dy |k_z|(x-y) |\phi_1 - \phi_2|^2(y)  \notag\\
		=&  \int_{\R^3} dy\|k_z\|_{L^1(\R^3)}|\phi_1 - \phi_2|^2(y) \int_Q dx   |k_z|(x-y)     \notag \\
			=&  \|k_z\|_{L^1(\R^3)}\int_{\R^3} dy|\phi_1 - \phi_2|^2(y) \|k_z\|_{L^1(Q-y)}    \notag \\
				\le & \|k_z\|_{L^1(\R^3)}\, \sup_{\ell \in \lat }\|\phi_1 - \phi_2\|_{L^2(Q+\ell)}^2 \, \sum_{\ell \in \lat} \sup_{x\in Q}\|k_z\|_{L^1(Q-\ell-x)} \notag \\
		\le & 8\|k_z\|_{L^1(\R^3)}^2\, \sup_{\ell \in \lat }\|\phi_1 - \phi_2\|_{L^2(Q+\ell)}^2,\label{eqn:Ndiff3}
\end{align}
where  $\lat$ is the Bravais lattice as in \eqref{eq:lat}. And in the last inequality for \eqref{eqn:Ndiff3} we used the estimate 
\[\sum_{\ell \in \lat} \sup_{x\in Q}\|k_z\|_{L^1(Q-\ell-x)}\le \sum_{\ell \in \lat} 2^3\|k_z\|_{L^1(Q-\ell)}\le 2^3\|k_z\|_{L^1(\R^3)} \]
since for $x\in Q$, $Q-\ell-x$ can be covered by at most $2^3$ translated cubes of $Q-\ell$ in $\R^3$. 

 Therefore,
\begin{equation}
	\|k_z*|\phi_1-\phi_2|\|_{L_\omega^\infty L^2(Q)}^2 \le  8\|\phi_1 - \phi_2\|_{L_\omega^\infty L^2_x(Q)}^2 \|k_z\|_{L^1(\R^3)}^2. 
\end{equation}
 
Combining \eqref{eqn:Ndiff2} and \eqref{eqn:Ndiff3}, we arrive at
 \begin{equation}
	\| N_n(\phi_1) - N_n(\phi_2)\|_{L_\omega^\infty L^2(Q)} \leq 8n \|\phi_\#\|_{L_\omega^\infty L^\infty_x(Q)}^{n-1} \|\phi_1 - \phi_2\|_{L_\omega^\infty L^2_x(Q)} \left|\oint\right|  \|k_z\|_{L^1(\R^3)}.\label{eqn:Ndiff4}
\end{equation} 
Finally, by definition of $k_z$ below equation \eqref{eqn:Ndiff1} and repeated applications of Young's inequality, we note that
\begin{equation}
	\|k_z\|_{L^1(\R^3)} \leq  \|g_z\|_{L^2(\R^3)}^2\|g_z\|_{L^1(\R^3)}^{n-1}.
\end{equation}
Define
\begin{equation}
 {C_{\beta, \mu,n} } := \Big(8n\left|\oint\right| \|g_z\|_{L^2(\R^3)}^2\|g_z\|_{L^1(\R^3)}^{n-1} \Big)^{1/n}< \infty.
\end{equation}
Recall the definition of $\left|\oint\right|$ in the \eqref{eqn:oint-def} and the exponential decay of the  Fermi-Dirac $f_{FD}$ as $Re z\to \infty$. We note that $\left|\oint\right| \|g_z\|_{L^2(\R^3)}^2\|g_z\|^{n-1}_{L^1(\R^3)}$ is integrable with respect to $z$. Therefore,   $C_{\beta,\mu,n}$ is bounded for given $\beta$ and $\mu$ fixed, uniformly in $n$ (since there is a $1/n$-th power). Thus, \eqref{eqn:Ndiff4} and the Sobolev inequality show 
 \begin{equation}
	\| N_n(\phi_1) - N_n(\phi_2)\|_{L_\omega^\infty L^2(Q)} \leq C_{\beta,\mu, n}^n \left(\|\phi_1\|_{L_\omega^\infty H^2_x(Q)}^{n-1} + \|\phi_2\|_{L_\omega^\infty H^2_x(Q)}^{n-1}\right)\|\phi_1 - \phi_2\|_{L_\omega^\infty H^2(Q)} .  \label{eqn:Ndiff5}
\end{equation} 
By \eqref{eqn:NexpandinNn}, \eqref{eqn:Ndiff5}, and the assumption that   $\|\phi_\#\|_{L^\infty_\om H^2_x} < \frac{1}{10}\sup_n C_{\beta, \mu, n}^{-1}$  for $\# = 1$ and $2$, we complete the proof of Equation \eqref{eqn:nonlin} of Theorem \ref{thm:nonlin}. 
\end{proof}

\appendix
\section{ {Stationary norms  }} \label{app:per-vol-set-up}
We briefly outline properties of the norms used in this article and provide a definition to $\den$ for the sake of completeness. A more in depth study can be found in Sections 2 and 3 of \cite{CLL}. Let $\lat$ denote a Bravais lattice in $\R^3$ and $Q$ its fundamental domain (for example, the Wigner-Seitz cell).  We will often suppress the dependence on $\om \in \Om$ for notation clarify (where $\Om$ is the probability space).  For $\ell \in  \R^3$ and $f_\om(x)=f(\om,x)$ a measurable function on $\Om \times Q$, given a 
measure preserving $\lat$-action $\tau$ on $\Om$, let $U_{\ell}$ denote the 
translation operator
\begin{equation}
    (U_{\ell} f_\om)(x) := f_{\tau_{\ell} \om}(x-\ell ).
\end{equation}
Recall from definition \ref{def:stationary-func} that a function $f \in L^p_\om L^q_x$ is said to be ($\mathcal{L}$) stationary if for $\ell\in \mathcal{L}$, 
\begin{equation}
    U_{\ell}f = f.
\end{equation}

\begin{lemma}  \label{lem:indep-of-loc-of-Q}
Let $f \in L^p_\om L^q_x$ and $1 \leq p, q \leq \infty$, the $L^p_\om L^q_x$ norm is independent of the location of the fundamental domain $Q$ translated by $\ell\in \mathcal{L}$.
\end{lemma}
\begin{proof}
Let $\tau$ denote the 
measure preserving map on $\Om$. It suffices to note that $h(\om) = \|f_\om\|_{L^q(Q)}$ satisfies:
\begin{equation}
    \|f_{\tau_{\ell} \om}\|_{L^q(Q)} = \|f_\om\|_{L^q(T_\ell  Q)},
\end{equation}
where $T_\ell (x) = x+\ell$ is the shift by $\ell\in \mathcal{L}$. Since $\tau_{\ell}$ is 
measure preserving, we see that
\begin{equation}\label{eq:app1}
    \E \|f_\om\|_{L^q(Q)}^p = \E \|f_\om\|_{L^q(T_\ell  Q)}^p
\end{equation}
for any $\ell \in  \mathcal{L}$. 
\end{proof}

A random operator,  $A$, on $L^2(\R^3)$ is said to be ($\lat$) stationary if $A$ commutes with $U_{\ell}$ for all $\ell \in  \R^3$.
Let $\Tr$ denote the usual trace on $L^2(\R^3)$. Given a random operator $A$, its density $\den A$ is a measurable function on $\Om \times Q$, if it exists, defined via the Riesz representation theorem and the formula
\begin{equation}
    \E \Tr f A = \E \int_{\R^3} f\den A \label{eqn:den-def},
\end{equation}
for all simple $C^\infty_c(\R^3)$-valued function $f$. That is, $f$ is of the form
\begin{equation}
    f = \sum_\alpha I_{S_\alpha} f_\alpha
\end{equation}
where the sum is a finite sum over $\alpha$, $I_{S_\alpha}$ is the indicator function of some measurable set $S_\alpha \subset \Om$, and $f_\alpha$ is a $C^\infty_c(\R^3)$ function. Moreover, the $f$ on the left hand side of \eqref{eqn:den-def} is regarded as a multiplication operator on $L^2(\R^3)$. When $A$ has an integral kernel $A(x,y)$ a.s., that is
\begin{equation}
    (Af)(x) = \int_{\R^3} A(x,y)f(y) dy,
\end{equation}
then,
\begin{equation}
    (\den A)(x) = A(x,x),
\end{equation}
whenever $A(x,x)$ is defined and unambiguous.

We outline a few special cases used in the paper where $\den A$ is defined. Define trace per volume $Q$ via  
\begin{equation}
    \Tr_Q A := \frac{1}{|Q|}\Tr {\bf 1}_Q A {\bf 1}_Q,
\end{equation}
where ${\bf 1}_Q$ is the indicator function of $Q$. If $A$ is stationary, $\Tr_Q$ is independent of translates of $Q$. For $1 \leq p, q \leq \infty$, there is a family of stationary Schatten spaces $L_\om^q\mathfrak{S}^p$ associated to $\Tr_Q$, given via the completions of
\begin{equation}
    L_\om^q\mathfrak{S}^p = \{ A_\om \in \mathcal{B}(L^2(\R^3)) \text{ and stationary} : \|A_\om\|_{L_\om^q\mathfrak{S}^p} < \infty \}
\end{equation}
with the norm  
\begin{equation} 
    \|A\|_{L_\om^q\mathfrak{S}^p}^q = \E \left(\Tr_Q (A^*A)^{p/2}\right)^{q/p}.
\end{equation} 
If $p$ or $q=\infty$, the usual sup norm is assumed. We remark that the usual (nonrandom) Schatten norm for operators on $X \subset \R^3$ is given as
\begin{equation}
    \|A\|_{\mathfrak{S}^p(X)}^p = \Tr {\bf 1}_X (A^*A)^{p/2} {\bf 1}_X,
\end{equation}
where ${\bf 1}_X$ is the indicator function on $X$. 

Let $R = (1-\Delta)^{-1}$. We have the following result. 
\begin{lemma} \label{lem:den-to-schatten}
Suppose that $A \in L_\om^\infty\mathfrak{S}^2$, then $\den  AR$ and $\den  RA \in L^\infty_\om L^2_x$. Moreover,
\begin{equation}
    \|\den  AR \|_{L^\infty_\om L^2_x}, \|\den R A\|_{L^\infty_\om L^2_x}  \lesssim \|A\|_{L_\om^\infty\mathfrak{S}^2}. \label{eqn:den-bound}
\end{equation}
\end{lemma}
 
\begin{proof}
We prove the case for $RA$ only. The case for $AR$ is treated similarly.  
We use the $L^2(Q)$-$L^2(Q)$ duality. Let $\phi \in L^2(Q)$ and apply H\"{o}lder's inequality to 
\begin{equation} 
    \Tr {\bf 1}_Q (\phi R) A {\bf 1}_Q \leq \|\varphi R\|_{\mathfrak{S}^2(\R^3)} \| A {\bf 1}_Q\|_{\mathfrak{S}^2(\R^3)}.
\end{equation}
Since $\frac{3}{2} < 2$ and we are in dimension 3, the Kato-Seiler-Simon inequality shows
\begin{equation}
    \Tr_{Q} [(\phi R) A] \lesssim \frac{1}{|Q|}\|\phi\|_{L^2(Q)} \| A {\bf 1}_Q\|_{\mathfrak{S}^2(\R^3)}.
\end{equation}
Since
\begin{equation}
    \frac{1}{|Q|}\| A {\bf 1}_Q\|_{\mathfrak{S}^2(\R^3)}^2 = \frac{1}{|Q|}\Tr {\bf 1}_Q A^* A {\bf 1}_Q = \|A\|_{ \mathfrak{S}^2(Q)}^2,
\end{equation}
we see that
\begin{equation}
\Tr_{Q} [(\phi R) A] \lesssim \|\phi\|_{L^2(Q)} \|A\|_{\mathfrak{S}^2(Q)}.
\end{equation}
The $L^2(\Om)$-$L^2(\Om)$ duality and the Riesz representation theorem show that there is some $\den RA \in L^2(Q)$ such that
\begin{equation}
    \frac{1}{|\Om|}\int_Q \varphi \den (RA) = \Tr_Q \varphi RA \label{eqn:denRA-est-1}
\end{equation}
for all $\varphi \in L^2(Q)$ and
\begin{equation}
    \|\den RA\|_{L^2(Q)} \leq \|A\|_{\mathfrak{S}^2(Q)}. \label{eqn:denRA-est-2}
\end{equation}
Since we have not specified the dependence of $A$ on $\om$, we see that \eqref{eqn:denRA-est-1} and \eqref{eqn:denRA-est-2} hold for all $\om$ a.s.. In particular, $\den RA_\om$ is well defined for a.e. $\om$. Taking $\sup_{\om \in \Om}$ in \eqref{eqn:denRA-est-2} in, we see that \eqref{eqn:den-bound} is proved.
\end{proof}

\Addresses


\begin{thebibliography}{}
\bibitem{AACan} A. Anantharaman, and \'E. Canc\`es (2009), Existence of minimizers for Kohn--Sham models in quantum chemistry. Ann. Inst. H. Poincar\'e Anal. Non Lin\'eaire {\bf 26}, 2425--2455.

\bibitem{CL}
R. Carmona, and J. Lacroix (2012), Spectral theory of random Schr\"odinger operators. Springer Science \& Business Media.

\bibitem{CDL}
\'E. Canc\`es, A. Deleurence, and M. Lewin (2008), A new approach to the modeling of
local defects in crystals: The reduced Hartree-Fock case. Comm. Math. Phys.  {\bf 281}, 129–177.

\bibitem{CLL}
\'E. Canc\`es, S. Lahbabi, and M. Lewin (2013),  Mean-field models for disordered crystals.  J. Math. Pures Appl. {\bf 100(2)}, 241-274.

\bibitem{CSL}
\'E. Canc\`es, G. Stoltz, and M. Lewin (2006), The electronic ground-state energy problem: A new reduced density matrix approach. J. Chem. Phys. {\bf 125(6)}, 064101.

\bibitem{CLeBL}  
I. Catto, C. Le Bris, and P.-L. Lions (2001), On the thermodynamic limit for Hartree-Fock type models. Ann. Inst. H. Poincar\'e Anal. Non Lin\'eaire {\bf 18}, Issue 6, 687 - 760.

\bibitem{CLeBL2}  
I. Catto, C. Le Bris, and P. -L. Lions (2002), On some periodic Hartree-type models for crystals. Ann. Inst. H. Poincar\'e Anal. Non Lin\'eaire {\bf 19}, Issue 2, 143 - 190. 



\bibitem{CS} I. Chenn, and I. M. Sigal (2020), On Derivation of the Poisson–Boltzmann Equation. J. Stat. Phys.  {\bf 180}, 954–1001.

\bibitem{CS1} I. Chenn, and I. M. Sigal (2019), On Effective PDEs of Quantum Physics. In: D'Abbicco M., Ebert M., Georgiev V., Ozawa T. (eds) New Tools for Nonlinear PDEs and Application. Trends in Mathematics. Birkh\"auser, Cham.


\bibitem{DS} M. Duerinckx, and C. Shirley (2021), A new spectral analysis of stationary random Schr\"odinger operators. Journal of Mathematical Physics, 62(7), 072106.
 

\bibitem{ELu} W. E, and J. Lu (2013), The Kohn-Sham equation for deformed crystals. Mem. Amer. Math. Soc. {\bf 221}, 1040. 






\bibitem{KS} W. Kohn, and L. J. Sham (1965), Self-consistent equations including exchange and correlation effects. Physical Review. {\bf 140} (4A), A1133-A1138.


\bibitem{Kir} W. Kirsch (2008), 
  An invitation to random Schr\"odinger operators. 
With an appendix by Fr\'ed\'eric 
Klopp. Panor. Synth\`eses, 25, Random Schr\"odinger operators, 1--119, Soc. Math. France, Paris. 



\bibitem{La} S. Lahbabi (2014), The reduced Hartree–Fock model for short-range quantum crystals with nonlocal defects.  Annales Henri Poincar\'e. Vol. 15. No. 7. Springer Basel. 


\bibitem{Lev} A. Levitt (2020), Screening in the finite-temperature reduced Hartree-Fock model. Arch. Ration. Mech. Anal.  {\bf 238.2}, 901-927.


\bibitem{Levy} M. Levy (1979), Universal variational functionals of electron densities, first order density matrices, and natural spin-orbitals and solutions of the $v$-representability problem. Proc. Natl. Acad. Sci. USA {\bf 76}, 6062 - 6065.


\bibitem{Levy2} M. Levy (1982), Electron densities in search of Hamiltonians. Phys. Rev. A {\bf 26}, 1200 - 1208.


\bibitem{Lieb3} E. H. Lieb  (1981), Thomas-Fermi and related theories of atoms and molecules. Rev. Modern Phys. {\bf 53}, no. 4, 603–641.

\bibitem{Lieb} E. H. Lieb (1983), Density Functionals for Coulomb Systems. Int. J. Quantum Chem. {\bf  24}, Issue 3, 243 - 277.

\bibitem{LS}
M. Lewin, and J. Sabin (2015), The Hartree equation for infinitely many particles I. Well-posedness theory. Comm. Math. Phys. {\bf 334(1)}, 117-170.

\bibitem{Nier} F. Nier (1993), A variational formulation of Schr\"{o}dinger-Poisson systems in dimension $d \leq 3$. Comm. Partial Differential Equations {\bf 18} (7 and 8), 1125-1147.  


 
\bibitem{PN} E. Prodan, and P. Nordlander (2003). On the Kohn-Sham equation with periodic back-
ground potential. J. Stat. Phys. {\bf 111}, Issue 3-4, 967-992.




\end{thebibliography}
\end{document}